\documentclass[10pt,letterpaper]{article}
\usepackage[top=1in, bottom=1in, left=1in, right=1in]{geometry}
\usepackage[varqu]{zi4}
\usepackage[square,sort&compress,numbers,sectionbib]{natbib}
\usepackage{graphicx}
\usepackage{cancel}
\usepackage{amsmath,amsfonts,amssymb,bbm,bm}
\usepackage{xspace}
\usepackage{enumitem}

\usepackage{amsthm}
\usepackage{booktabs} 
\usepackage{tikz-cd}
\usepackage{xcolor}
\usepackage{nameref}
\usepackage{thmtools} 
\usepackage{thm-restate}
\usepackage[linktocpage=true,pagebackref=true,colorlinks,bookmarks,citecolor=blue,linkcolor=black,bookmarksopen,bookmarksnumbered]
{hyperref}
\usepackage{cleveref}
\usepackage{authblk}
\RequirePackage[libertine,vvarbb]{newtxmath} % math font for libertine
\RequirePackage{bm}
\usepackage[linesnumbered,ruled,noend]{algorithm2e}

\SetCommentSty{mycommfont}

\usepackage{subfig}
\usepackage{enumitem}
\SetKwInput{KwAssume}{Assumptions}
\SetKwInput{KwInput}{Input}
\SetKwInput{KwOutput}{Output} 
\SetKwInput{KwPre}{Precondition}
\SetKwInput{KwInv}{Invariant}
\SetKwFunction{MainAlgorithm}{{\tt PROP1\text{+}1/2MMS}}
\SetKwFunction{BipartiteConstruction}{{\texttt BipartiteConstruction}} 
\newcommand{\abs}[1]{\left\vert#1\right\vert}
\newcommand{\MMS}{\mathrm{MMS}}

\DeclareMathOperator*{\argmin}{argmin}
\definecolor{darkgreen}{rgb}{0.0, 0.7, 0.0}
\definecolor{darkred}{RGB}{220, 32, 32}

\newcommand{\ItemPartition}{{\tt{MinBundleKept}}\xspace}
\newcommand{\PropOracle}{{\tt {PROP1\text{-}nonPROP}}\xspace}

\theoremstyle{definition}
\newtheorem{definition}{Definition}
\theoremstyle{plain}
\newtheorem{theorem}{Theorem}
\newtheorem{lemma}{Lemma}
\newtheorem{example}{Example}

\usetikzlibrary{fit,calc, patterns}
\newcommand{\nonl}{\renewcommand{\nl}{\let\nl\oldnl}}
\colorlet{mygray}{gray!40}
\colorlet{myblue}{cyan!60}
\colorlet{myPink}{blue!70}

\title{Logarithmic Comparison-Based Query Complexity for Fair Division of Indivisible Goods\thanks{The paper has been accepted for publication in the Conference on Web and Internet Economics (WINE 2024).
The work was done when Jiaxin was an undergraduate student at Shanghai Jiao Tong University.}}
\author[1]{Xiaolin Bu}
\author[2]{Zihao Li}
\author[3]{Shengxin Liu}
\author[4]{Jiaxin Song}
\author[1]{Biaoshuai Tao}
\affil[1]{Shanghai Jiao Tong University, \{lin\_bu, bstao\}@sjtu.edu.cn}
\affil[2]{Nanyang Technological University, zihao004@e.ntu.edu.sg}
\affil[3]{Harbin Institute of Technology, Shenzhen, sxliu@hit.edu.cn}
\affil[4]{University of Illinois, Urbana-Champaign, jiaxins8@illinois.edu}

\date{}

\begin{document}
\maketitle

\begin{abstract}
We study the problem of fairly allocating $m$ indivisible goods to $n$ agents, where agents may have different preferences over the goods.
In the traditional setting, agents' valuations are provided as inputs to the algorithm.
In this paper, we adopt the query model, which has been widely considered for other similar problems (such as matching~\cite{nisan2021demand}, graph isomorphism~\cite{onak2018isomorphism}, and equilibrium in game~\cite{babichenko2016equilibria}), and apply it to the fair division problem. 
In particular, we consider a new \emph{comparison-based query model}, where the algorithm presents two bundles of goods to an agent and the agent responds by telling the algorithm which bundle she prefers.
We investigate the query complexity for computing allocations with several fairness notions including \emph{proportionality up to one good} (PROP1), \emph{envy-freeness up to one good} (EF1), and \emph{maximin share} (MMS).
Our main result is an algorithm that computes an allocation that satisfies both PROP1 and $\frac12$-MMS within $O(\log m)$ queries with a constant number of $n$ agents.
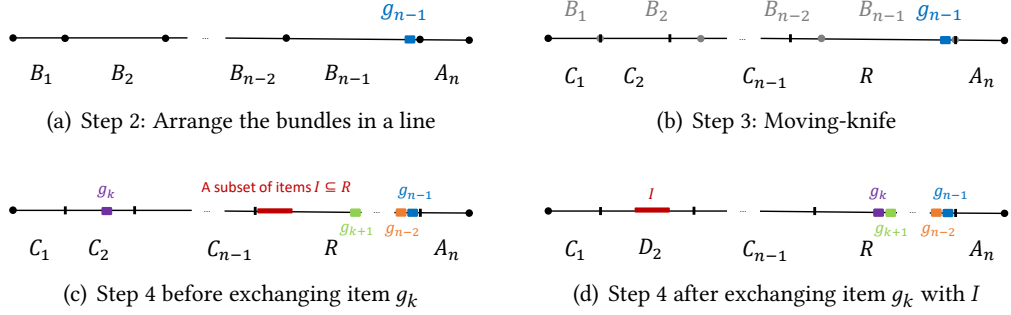
For identical and additive valuation, we present an algorithm for computing an EF1 allocation within $O(\log m)$ queries with a constant number of $n$ agents.
To complement the positive results, we show that the lower bound of the query complexity for any of the three fairness notions is $\Omega(\log m)$ even with two agents.
\end{abstract}

\section{Introduction}

\emph{Fair division} of indivisible resources concerns how to fairly allocate a set of heterogeneous indivisible goods/items to a set of agents with varying preferences for these resources~\cite{AmanatidisAzBi23}.
The concept of fair division was initially introduced by Steinhaus~\citep{Steinhaus48,Steinhaus49} and has since been extensively studied in the fields of economics, computer science, and social science.
The primary objective of fair division is to achieve \emph{fairness} given that agents have different preferences for the resources.

\paragraph{Query-based fair allocation.}
In most existing literature, agents' preferences are typically represented by a set of \emph{valuation functions}, where each function assigns a non-negative value to every subset of items, reflecting how much an agent values it.
In the classical setting~\cite{budish2011combinatorial,procaccia2014fair,ghodsi2021fair,bu2022complexity,AkramiGa24}, a fair division algorithm assumes explicit knowledge of each agent's valuation function: these valuation functions are provided as inputs to the algorithm. This is referred to as \emph{the direct revelation}.

The direct revelation simplifies the theoretical analysis.
However, many real-world scenarios may not fit this neat model since agents' valuation functions may be unavailable upfront.
For example, to enhance user experience, many platforms adopt indirect and gentle approaches to understand users' preferences. 
Ridesharing platforms (e.g., Lyft and Uber) collect user ratings after each ride to gain insights into users' preferences for different types of cars.
Similarly, online homestay platforms (e.g., Airbnb and Ctrip) use recommendation systems to suggest rooms on their main pages, learning users' preferences by analyzing their clicks and interaction histories.

To model these scenarios, a line of works~\cite{lipton2004approximately,plaut2020almost,Oh_2021,LiMSS24} assumes that the allocation algorithm does not have explicit knowledge of agents' preferences.
Instead, a \emph{value-based query} is adopted to access an agent's value for a specific bundle of items.
This line of research focuses on optimizing the query complexity to compute a fair allocation.

\paragraph{Comparison-based queries.}
In this paper, we propose a new \emph{comparison-based query model}.
In this model, agents retain their \emph{intrinsic} valuation functions.
However, instead of requiring them to disclose their valuation functions directly or querying their explicit value for specific bundles, we learn about each agent's preferences by querying the value relationship among different bundles.
In each comparison-based query, we present two bundles to an agent and ask her to specify which bundle she prefers.
Under this model, we study the problem of optimizing query complexity to obtain a fair allocation.

Compared to the value-based query model, the comparison-based query model offers several advantages in many practical scenarios.
It is significantly more convenient for agents to report their ordinal preferences rather than numerical values.
When providing feedback or expressing preferences, most users prefer a simplified process. 
The comparison-based query allows users to answer only with a simple ``yes'' or ``no'', without the need to disclose concrete values.
It also protects agents' privacy to a large extent.
From the perspective of information theory, each response conveys only one bit of information, which is the minimum.
Our algorithmic results based on this new query model achieve fair allocations while making agents reveal as little information as possible.

\paragraph{Fairness criteria.}
Various notions of fairness have been proposed over the years.
Among them, \emph{envy-freeness} (EF)~\cite{Foley67,varian1973equity} and \emph{proportionality} (PROP)~\cite{Steinhaus49} are among the most natural ones. EF requires that no agent envies another, while PROP ensures that each agent receives at least an average share of the total value.

Unfortunately, such allocations may not always exist for indivisible items (e.g., allocating a single item to two agents).
The arguably most prominent relaxation of envy-freeness is \emph{envy-freeness up to one item} (EF1)~\citep{lipton2004approximately,caragiannis2019unreasonable}, which states that for any pairs of agents $i$ and $j$, agent $i$ will not envy $j$ if one item is (hypothetically) removed from $j$'s bundle.
As for proportionality, two different relaxations have been proposed in the literature: \emph{proportionality up to one item} (PROP1)~\citep{ConitzerFrSh17,AzizCaIg22} and
%\emph{maximin share (MMS)~\citep{akrami2023breaking,budish2011combinatorial,procaccia2014fair,ghodsi2021fair,AkramiGaSh23,AmanatidisMaNi17,BarmanKr20,FeigeSaTa21,GargTa21,GargMcTa19}.
\emph{maximin share} (MMS)~\citep{budish2011combinatorial,procaccia2014fair}.
%\bt{Did we include the most recent paper on MMS? Or, shall we just cite the paper introducing MMS, and leave the others to the related work section?}\zihao{I have added this, and I think it is better to cite only the paper introducing mms here, and add a new section in related work}\zihao{and do we need to order these cited papers in lexicographical order?}\shengxin{I agree that we can just mention the very first papers on MMS and give some brief discussion in the related work section.}
Similar to EF1, PROP1 ensures proportionality for each agent $i$ if one item is (hypothetically) added to $i$'s bundle.
For MMS, each agent $i$ has a threshold $\MMS_i$ that is defined as the maximum value of the least preferred bundle in any $n$-partition of the goods (where $n$ is the number of the agents).
An allocation is MMS if each agent receives a value at least their MMS threshold.

For additive valuations, EF1 and PROP1 allocations are guaranteed to exist~\citep{lipton2004approximately, budish2011combinatorial}, while an MMS allocation may not always exist~\citep{procaccia2014fair}.
Therefore, researchers study the approximation version of MMS: $\alpha$-MMS, indicating each agent's received value is at least $\alpha\cdot$ $\MMS_i$.
The best-known result to date is that an $\alpha$-MMS allocation always exists for a constant $\alpha=3/4 + 3/3836$~\cite{AkramiGa24}.
Although an MMS allocation always satisfies PROP1 for additive valuations, there exists a constant $\alpha<1$ such that PROP1 and $\alpha$-MMS are incomparable.
When the goods have similar values, PROP1 is close to PROP, whereas $\alpha$-MMS, being close to $\alpha$-proportionality, is a weaker criterion.
On the other hand, in cases where a single good is worth more than $1/n$ of the total value, an allocation where an agent receives an empty bundle can still be PROP1. In such cases, $\alpha$-MMS provides a stronger fairness guarantee.

% In this paper, we design an algorithm that computes allocations that satisfy both PROP1 and $\frac12$-MMS.

\paragraph{Sublinear complexity in goods.} 
% \paragraph{Sublinear complexity on number of goods.}
A challenge in many fair division models is that the number of goods $m$ can be extremely large (e.g., computational resources, housing resources, transportation resources, etc).
%Hence, direct revelations of the values of all the items may be infeasible.
Our comparison-based query model offers a solution for handling such cases.
We aim to design algorithms requiring a number of queries that is sublinear in the number of items $m$.
Note that direct adaptations of existing algorithms~\citep{lipton2004approximately,Oh_2021,LiMSS24} require the query numbers that are linear in $m$ to compute an EF1 allocation. 
%Therefore, our results are also tight with respect to $m$.

%As presented in Section~\ref{sec:prop1+mms}, our main algorithm can compute an allocation satisfying both PROP1 and $\frac12$-MMS in $O(\mathrm{poly}(n){\cdot}\log m)$ queries, which is just $O(\log m)$ for a constant number of agents, which is much less than the input length of the original direct revelation model.
%Moreover, in Corollary~\ref{thm:prop1_lower_bound}, we provide a lower bound of $\Omega(\log m)$ for even two agents.
%Therefore, our results are also tight with respect to $m$.

\subsection{Our Results and Technical Overviews}
In this paper, we focus on additive valuations and deterministic algorithms.
As the main result of the paper, we propose an algorithm that uses $O(n^4\log m)$ comparison-based queries to compute an allocation that is simultaneously PROP1 and $\frac12$-MMS.
%Moreover, we provide a lower bound of $\Omega(\log m)$ for any fixed number of agents.
For EF1, we present algorithms that use $O(\log m)$ queries for up to three agents and $O(n^2\log m)$ queries for $n$ agents with identical valuations.
Additionally, the results are complemented by the fact that $\Omega(\log m)$ queries are necessary to achieve any of the three mentioned fairness notions -- PROP1, EF1, and $\alpha$-MMS (for any $\alpha>0$) -- even for any fixed constant number of agents with identical valuations.
This lower bound matches the query complexity of our algorithms for constant numbers of agents.

Our main result is built upon three techniques: the subroutine \ItemPartition, the subroutine \PropOracle, and the \emph{Hall matching}. 
The \ItemPartition subroutine aims to find a partial allocation that is approximately balanced, with most of the goods allocated. 
It iteratively finds a tentative $n$-partition, finds the bundle with the minimum value using comparison-based queries, and keeps this bundle.
This subroutine serves as the key building block in all our algorithms.
Using \ItemPartition, we can design algorithms for computing PROP1 and EF1 allocations when agents' valuations are identical.

Roughly speaking, the \PropOracle subroutine decides whether a bundle approximately satisfies the proportional requirement of an agent.
Specifically, as discussed at the beginning of Sect.~\ref{sec:prop1}, it is impossible to verify with comparison-based queries whether an allocation satisfies the PROP1 condition for an agent or whether a bundle is worth at least $1/n$ of the agent's value of the entire item set (it is perhaps interesting to see that computing a PROP1 allocation is possible while checking whether an allocation is PROP1 is not).
The \PropOracle subroutine takes a bundle and an agent as inputs, and outputs ``yes'' if the bundle satisfies the PROP1 requirement for the agent, and ``no'' if the bundle fails to achieve PROP.
Notice that a bundle may satisfy both ``yes'' and ``no'' conditions; in this case, either outputs ``yes'' or ``no'' is acceptable.

With the subroutine \PropOracle, given $n$ bundles, we can construct a bipartite graph $G=(A,B,E)$ where $A$ contains $n$ vertices representing the agents, $B$ contains $n$ vertices representing the bundles, and an edge indicates a bundle surely satisfies the PROP1 requirement for an agent.
To find a PROP1 and $\frac12$-MMS allocation, we iteratively find a matching from $X\subseteq A$ to $Y\subseteq B$ such that no edge exists from an agent in $A\setminus X$ to a bundle in $B$. 
We then allocate the bundles in $Y$ to the agents in $X$.
This process guarantees that the agents in $A\setminus X$ believe the remaining unallocated items (i.e., the union of the bundles in $B\setminus Y$) are sufficiently valuable, guaranteeing that they will eventually receive a bundle satisfying the PROP1 and $\frac12$-MMS conditions.
Such a matching can be found by removing all agents in $A$ that fail the Hall condition $|\Gamma(S)|\geq |S|$, and we call this matching from $X$ to $Y$ as a \emph{Hall matching}.

\begin{figure}
    \centering
    \begin{tikzpicture}[
        every node/.style={draw, rounded corners, minimum width=3cm, minimum height=1cm, text width=3cm, align=center},
        every path/.style={thick, -{Stealth[round]}}
    ]

    \node[draw=none, font=\bfseries] (techniques) at (-4, 2) {Techniques};
    \node[draw=none, font=\bfseries] (results) at (-4, -1) {Results};

    \scriptsize
    \node (itemPartition) [fill=gray!15] at (0,2) {\ItemPartition};
    \node (propOracle) [fill=gray!15] at (4, 2) {\PropOracle};
    \node (hallMatching) [fill=gray!15] at (8, 2) {\emph{Hall Matching}};
    
    \node (ef1Identical) [fill=blue!10] at (0,0) {Identical EF1 \\ (\Cref{thm:ef1_identical})};
    \node (prop1Identical) [fill=blue!10] at (4, 0) {Identical PROP1 \\ (\Cref{thm:prop1_identical})};
    
    \node (prop1General) [fill=blue!10] at (8,0) {PROP1 \\ (\Cref{thm:prop1_general})};
    \node (prop1MMS) [fill=blue!10] at (4, -2) {\scriptsize PROP1 + $\frac12$MMS \\ (\Cref{thm:prop1_mms_general})};
    
    \path (itemPartition) edge (ef1Identical)
                          edge (prop1Identical);
    \path (propOracle) edge (prop1General)
                      edge[bend right=65] (prop1MMS);
    \path (hallMatching) edge (prop1General)
                         edge (prop1MMS);
    \path (ef1Identical) edge (prop1MMS);
    \path (prop1Identical) edge (prop1General);
    
    \end{tikzpicture}
    \caption{Overview of our main techniques and results, where an arrow represents the former technique/result is applied in the latter result.}
    \label{fig:technical_overview}
\end{figure}

The relationship between our techniques and results is shown in Fig.~\ref{fig:technical_overview}.
In Sect.~\ref{sec:prop1}, we present an algorithm with $O(n^4\log m)$ queries to find a PROP1 allocation.
In Sect.~\ref{sec:ef1}, we discuss our results on EF1 allocations.
Finally, in Sect.~\ref{sec:prop1+mms}, we present our main result: an algorithm with $O(n^4\log m)$ queries that finds an allocation satisfying PROP1 and $\frac12$-MMS simultaneously.
Note that although our result for PROP1 in Sect.~\ref{sec:prop1} is weaker than our main result in Sect.~\ref{sec:prop1+mms}, and the latter does not rely on the former, we discuss the PROP1 algorithm first, as it provides a simpler illustration of our three techniques.

\subsection{Related Work}
\label{sec:related_work}
\paragraph{Value-based query complexity for indivisible items.}
There have been several studies on value-based query complexity for computing fair allocations of indivisible items~\cite{plaut2020almost,Oh_2021}, which allows the algorithm to acquire agent $i$'s specific value for a bundle $X$ of items, i.e., $u_i(X)$.
Observe that each comparison-based query can be simulated by two value-based queries, which implies that the lower bound of comparison-based query complexity is no less than that of value-based query complexity asymptotically.
For value-based queries, \cite{plaut2020almost} showed that, even for two agents with identical and monotone valuations, the query complexity of computing an EFX (envy-free up to any item) allocation is exponential in terms of the number $m$ of indivisible items.
\cite{Oh_2021} explored the value-based query complexity for more fairness notions.
For EFX with two agents, \cite{Oh_2021} provided a lower bound of $\Omega(m)$ and a matching upper bound of $O(m)$ queries for additive valuations.
For EF1 with no more than three agents, the query complexity is $\Theta(\log m)$ for additive valuations and $O(\log^2 m)$ for monotone valuations. 
With more than three agents, a careful analysis of the classic \emph{envy-graph} procedure is provided, which requires $O(nm)$ queries in the worst case.
Recently, in parallel to our work, \cite{LiMSS24} also studied comparison-based queries. Specifically, they (1) showed that the classic round-robin algorithm can be implemented in $O(nm \log(m/n))$ comparison-based queries in the worst case, and (2) proved that $\Omega(nm + m \log m)$ comparison-based queries are necessary to implement the round-robin algorithm. We note that the round-robin algorithm can be used to return an EF1 or PROP1 allocation; however, our focus is on designing algorithms that use a sublinear number of comparison-based queries.
%
% \shengxin{Please help check the discussion on Warut's recent SAGT paper~\cite{LiMSS24}.}

\paragraph{Communication complexity.}
When there is no restriction on the queries, the communication complexity (i.e., the number of bits exchanged during interaction) of finding a fair allocation has also attracted great interest in existing papers~\citep{plaut2020communication,BN19,feige2024low}.
Plaut and Roughgarden~\cite{plaut2020communication} considered the communication complexity for $\alpha$-EF/$\alpha$-PROP beyond additive valuations.
Our primary positive result directly implies that the communication complexity of finding a PROP1 and $\frac12$-MMS allocation requires at most $O(n^4\log m)$ bits.
Feige's recent work~\cite{feige2024low} further improved the upper bound of finding a PROP1 allocation to $O(n\log m \log n)$ and showed that this bound also applies for finding an APROP (PROP1 and $\frac{n}{2n-1}$-TPS) allocation.
However, the above results cannot yield a better upper bound of the query complexity for finding a PROP1+$\frac12$-MMS allocation within the comparison-based query model.
Specifically, the former result requires each agent to compute a PROP1 allocation offline according to her own valuation function (thus no information exchange is needed), whose complexity is not included in the communication model yet is considered in our query model, and the latter result requires value-based queries.
\cite{feige2024low} also provided a lower bound $\Omega(n\log \frac{m}{n})$ of the communication complexity (even for randomized protocols) for finding a PROP1 allocation, which implies the lower bound under our model for a constant number of agents.

\paragraph{Query complexity for divisible items.}
In the context of fair division of \emph{divisible} resources (i.e., cake-cutting), a valuation function can be a general real function on an interval, which makes it impossible to be succinctly encoded.
Thus, the query complexity has been extensively studied in the cake-cutting literature~\cite{brams1995envy,robertson1998cake,stromquist2008envy,kurokawa2013cut,aziz2015discrete,aziz2016discrete,procaccia2017lower,branzei2022query,cechlarova2012computability}.
The most well-known query model is the Robertson-Webb (RW) query model~\cite{robertson1998cake}, where two types of queries, namely \textbf{Eval} and \textbf{Cut}, are allowed.
In particular, $\textbf{Eval}_i(x, y)$ asks agent $i$ the value of interval $[x,y]$ and $\textbf{Cut}_i(x,r)$ asks agent $i$ for a point $y$ such that agent $i$'s value to interval $[x,y]$ is $r$.
The moving-knife~\cite{dubins1961cut} and Even-Paz~\cite{even1984note} algorithms compute a proportional allocation with connected pieces within the complexity of $O(n^2)$ and $O(n\log n)$ RW queries respectively.
On the other hand, a lower bound of $\Omega(n\log n)$ RW queries is given by \cite{edmonds2006cake}.
Computing an envy-free allocation is more challenging.
An envy-free allocation can be found by the \emph{cut-and-choose} protocol for two agents.
For a general number of agents, \cite{aziz2016discrete} proposed an algorithm within a finite number of RW queries while the best known lower bound is $\Omega(n^2)$ by~\cite{procaccia2009thou}.

\paragraph{Other query models for indivisible items.}
Our problem is also implicitly related to the fair division of indivisible goods using only ordinal or partial information about the preferences (e.g.,~\citep{HS21,aziz2023possible,benade2022dynamic}).
We refer to the excellent survey for an overview of recent progress on fair allocations of indivisible goods~\cite{AmanatidisAzBi23}.

\section{Preliminaries}\label{sec:preliminary}
Denote $\{1, \dots, k\}$ by $[k]$. 
There are $m\in \mathbb{N}^+$ indivisible goods/items, denoted by $M = \{g_1, \ldots, g_m\}$, required to be allocated to $n\in \mathbb{N}^+$ agents, denoted by $[n] =\{1,\ldots,n\}$.
A \emph{bundle} is a subset of $M$. An allocation $\mathcal{A}=\{A_1,A_2,\dots,A_n\}$ is a $n$-partition of the set of goods $M$, where $A_i$ is the bundle assigned to agent $i$.
We say an allocation is a \emph{partial allocation} if the union of these bundles is a proper set of $M$ (i.e., $\cup_{i=1}^n A_i\subsetneq M$).

Assume each agent $i \in [n]$ has an \emph{intrinsic} and \emph{non-negative} valuation function $u_i:\{0,1\}^{[m]}\rightarrow \mathbb{R}_{\ge 0}$.
The set of all valuation functions is called value profile $\mathcal{U}$, which is not given as an input.
To learn about an agent's preference, an algorithm is only allowed to select an agent, present her with two bundles, and ask which one she prefers. 
Note that even if we could get the comparison result, we still do not know the explicit values of the two bundles.
Our query model is referred to as the \emph{comparison-based query model}.
\begin{definition}
We adopt the following query in the comparison-based query model: 
\begin{itemize}
    \item $\textbf{Comp}_i(X,Y)$: given two bundles $X, Y \subseteq M$, ask agent $i$ compare $u_i(X)$ and $u_i(Y)$ and return the bundle with a higher value (if $u_i(X)=u_i(Y)$, then return an arbitrary one).
\end{itemize}
\end{definition}

Throughout this paper, we assume all the valuation functions are \emph{additive}, i.e., $u_i(S)=\sum_{g\in S}u_i(\{g\})$ for any $i\in [n]$ and $S\subseteq M$.
For simplicity, we refer to $u_i(\{g\})$ as $u_i(g)$.
We say the valuation functions are \emph{identical} if $u_1=u_2=\dots=u_n$.

To measure fairness, the most natural notions are \emph{envy-freeness} (EF) and \emph{proportionality} (PROP).
An allocation $\mathcal{A}$ is said to satisfy \emph{envy-freeness (EF)} if no agent prefers the bundle of another agent to her own, i.e., $u_i(A_i) \geq u_i(A_j)$ for all $i, j\in [n]$
An allocation $\mathcal{A}$ is said to satisfy \emph{proportionality (PROP)} if $u_i(A_i)\ge u_i(M)/n$ for any agent $i$.

Since EF and PROP may not always be satisfied, we consider the following relaxations, EF1 and PROP1.
For a verbal description, EF1 requires that no agent envies another agent following the removal of some single good from the other’s bundle.
PROP1 ensures that, for every agent $i$, there exists an item $g$ such that the total value of agent $i$'s allocation is at least $1/n$ of the overall value (i.e., $u_i(M)$) if $g$ is included in agent $i$'s bundle.
\begin{definition}[EF1]
An allocation $\mathcal{A}$ is said to satisfy \emph{envy-freeness up to one item (EF1)} if, for any two agents $i$ and $j$, either $u_i(A_i)\geq u_i(A_j)$ or there exists an item $g\in A_j$ such that $u_i(A_i)\ge u_i(A_j\setminus\{g\})$.
\end{definition}

\begin{definition}[PROP1]
An allocation $\mathcal{A}$ is said to satisfy \emph{proportionality up to one item (PROP1)}, if for any agent $i$, there exists an item $g\in M\setminus A_i$ such that $u_i(A_i\cup\{g\})\ge u_i(M) / n$.
\end{definition}

To find an EF1 allocation for monotonic valuation functions, Lipton et al.~\cite{lipton2004approximately} proposed \emph{envy-graph}, where each vertex represents an agent and there is a directed edge from vertex $i$ to vertex $j$ if agent $i$ envies agent $j$.
The \emph{envy-graph procedure} starts from an empty allocation.
When there is an unallocated item, it selects a source agent (an agent whose corresponding vertex in the envy-graph has no incoming edges) and assigns that item to her, and then runs the \emph{cycle-elimination} algorithm until there are no cycles in the envy-graph.
\begin{definition}[Cycle-elimination]
If there is a cycle $i_1\rightarrow i_2 \rightarrow \ldots i_k \rightarrow i_1$ (we identify $i_{k+1}=i_1$) on the envy-graph, shift the bundles along this path (i.e., agent $i_j$ receives agent $i_{j+1}$'s bundle for any $j\in [k]$).
\end{definition}

Besides, \emph{maximin share} (MMS) is also extensively studied for indivisible items.
The maximin share of agent $i$, denoted by $\MMS_i$, is defined as the maximum value of the least preferred bundle among all allocations under valuation function $u_i$.
By the definition of MMS, we can observe that $\MMS_i\le \frac{1}{n}u_i(M)$.
An allocation is said to be MMS-fair if the value of the bundle received by each agent $i$ is at least $\MMS_i$, which may not always exist~\cite{procaccia2014fair}. 
An allocation is said to be $\alpha$-maximin share fair ($\alpha$-MMS) if each agent $i$ receives a bundle with at least $\alpha$-fractional of their maximin share.

\begin{definition}[MMS]
Let $\Pi(M)$ be the set of all the possible allocations.
The maximin share of agent $i$ is defined as $\MMS_i = \max_{\mathcal{A}\in \Pi(M)} \min_{j\in [n]} u_i(A_j)$.
\end{definition}

\begin{definition}[$\alpha$-MMS]
An allocation $\mathcal{A}$ satisfies $\alpha$-maximin share fair ($\alpha$-MMS), if for any agent $i$, $u_i(A_i)\ge \alpha\cdot \MMS_i = \alpha \cdot \max_{\mathcal{A}\in \Pi(M)} \min_{j\in [n]} u_i(A_j)$.
\end{definition}

In this paper, we study the problem of optimal query complexity (defined as the \emph{minimum number of comparison-based queries} required) to find a fair allocation.
Note that the algorithms in this paper are restricted to be deterministic.
We first show a lower bound of the query complexity for finding allocations satisfying the above fairness notions.
As discussed in~\Cref{sec:related_work}, a comparison-based query can be implemented by two value-based queries.
This directly implies that the lower bound for comparison-based query complexity is at least half of the one for the value-based query.
In \Cref{thm:prop1_lower_bound}, using the similar construction of~\cite[Theorem 5.5]{Oh_2021}, we demonstrate that the value-based query complexity for finding either a PROP1, EF1 or $\alpha$-MMS allocation is $\Omega(\log m)$ when $n$ is given as a constant number, which further implies the comparison-based query complexity of finding these fair allocations is also $\Omega(\log m)$

\begin{theorem}\label{thm:prop1_lower_bound}
The comparison-based query complexity of computing a PROP1/EF1/$\alpha$-MMS allocation is $\Omega(\log m)$ even for a constant number of agents with identical valuations. Here, $\alpha$ is an arbitrary real number in $(0,1]$.    
\end{theorem}
\begin{proof}
Consider $n$ agents with identical and binary valuation functions (i.e., with $u_i(g)\in\{0,1\}$ for each $g\in M$).
For each agent, there exist $n+1$ items 
such that this agent has a value of $1$ for each of them and a value of $0$ for the remaining items.
In this instance, if an allocation is PROP1 or $\alpha$-MMS for any $\alpha>0$, each bundle should include at least one valued item and, at most, two valued items.

Next, we answer the comparison adversarially as follows:
Let $G$ be initialized as  $M$.
When the algorithm queries the utility of a bundle $H$ for agent $i$, we return $n+1$ if  $\abs{G \cap H} \ge \abs{G}/2$ and $0$ otherwise.
If $\abs{G \cap H} \ge \abs{G}/2$, then we update $G$ to $H$.
Otherwise, we update $G$ to $G\setminus H$.
In each query, the only information the allocator gets is that all the valued items are located in $G$.
When $\abs{G} > 2n$, there exists a bundle with at least three items from $G$ in any allocation.
Hence, PROP1 and $\alpha$-MMS cannot be guaranteed.
Thus, the algorithm must keep querying until $\abs{G} \le 2n$.
Thus, the minimum number of value-based queries should be at least $\log\left(m/2n\right)$, which means the value-based query complexity is $\Omega(\log m)$ since $n$ is assumed to be much smaller than $m$.
As EF1 is stronger than PROP1, the lower bound for finding an EF1 allocation also holds.
\end{proof}

\section{Query Complexity of PROP1}\label{sec:prop1}
In this section, we focus on the query complexity of finding a PROP1 allocation.
A straightforward idea is to adapt the moving-knife method~\cite{stromquist1980cut,budish2011combinatorial,BILO2022197}, which requires approximately $O(\mathrm{poly}(n){\cdot}\log m)$ queries.
However, in our model, determining whether a bundle is PROP1 (or PROP) for an agent is impossible.
In \Cref{sec:prop1_oracle}, we provide two counter-examples and introduce the subroutine \PropOracle.
Next, in \Cref{sec:prop1_identical_val}, we show that a PROP1 allocation can be found in $O(n^2\log m)$ queries under identical valuations.
Finally, by applying the \PropOracle subroutine with the result from \Cref{sec:prop1_identical_val}, we demonstrate that even under non-identical valuations, a PROP1 allocation can be found in $O(n^4\log m)$ queries in \Cref{sec:prop1_nonidentical_val}.

\subsection{Subroutine for Determining PROP1 and non-PROP}
\label{sec:prop1_oracle}
Below, we provide two counter-examples to illustrate that it is not always possible to determine whether an allocation is PROP1 (or PROP) using the comparison-based query.
\begin{example}[Impossibility for PROP1]
Consider an instance with $m=n+1$ items.
Let $x$ be an unknown number such that $x\in (1,2)$.
For each $i\in [n]$, $u_i(g_1)=\ldots =u_i(g_n) =1$ and $u_i (g_{n+1}) =x$.
To determine whether an empty bundle $\emptyset$ is PROP1 for an agent $i$, we need to verify whether the unknown value $x$ satisfies $x \ge \frac{n+x}n$, which is equivalent to $x \geq \frac{n}{n-1}$.
However, this condition cannot be determined by comparison-based queries.
\end{example}

\begin{example}[Impossibility for PROP]
Consider an instance where $n = m =3$ and $u_i(g_1)= 0, u_i(g_2) = x, u_i(g_3) = 1$, where $0 < x< 1$.
To determine whether $\{g_2\}$ is PROP, one must figure out whether $x\ge 1/2$, which is still impossible to be determined by our query model.
\end{example}

\begin{algorithm}[h]
\caption{\PropOracle$(n, M, B, u)$}
\label{alg:oracle}
Line up the items in the order $(B, M\setminus B)$\;
Let $R\leftarrow M\setminus B$ be the remaining items\;
\For{$i= 1, \dots, n-1$}{
Find bundles $D_i$ and $I_i$ from the remaining item set $R$ such that $u(D_i)\le u(B)$ and $u(D_i\cup I_i)>u(B)$, where $I_i$ is a singleton\;
\If{$D_i$ and $I_i$ do not exist} {
Let $D_i\leftarrow R$ and $I_i \leftarrow \emptyset$\;
Update $R\leftarrow \emptyset$\;
\textbf{break}\;
}
Update $R\leftarrow R\setminus (D_i\cup I_i)$\;
}
\If{$R$ is empty}{
\Return $yes$\;
}
\Return $no$\;
\end{algorithm}

To tackle this issue, we introduce the subroutine, \PropOracle, as shown in Algorithm~\ref{alg:oracle}, which allows the overlapping of ``yes'' bundles (satisfying PROP1) and ``no'' bundles (not satisfying PROP).
Given a bundle $B\subseteq M$ and an underlying valuation $u(\cdot)$, which can be accessed through comparison-based queries, the subroutine verifies whether $u(B)<1/n$ or $u(B\cup\{g\})\ge1/n$ for some item $g$ outside $B$.
We first partition the items into a sequence $B, D_1,I_1,\ldots, D_{n-1}, I_{n-1}, R$ such that $\abs{I_i}\le 1$, each bundle $D_i$ is weakly less valuable than $B$, and each bundle $D_\ell\cup I_i$ is more valuable than $B$.
Some of the latter items and bundles, however, may be empty due to the lack of items.
This does not affect the correctness of our algorithm, and we can directly ignore the desired constraints for these bundles.
If there are no remaining items in the sequence, i.e., $R=\emptyset$, this implies that PROP1 and the algorithm will output ``yes''.
Otherwise, it will output ``no'', indicating $B$ is non-PROP.
Note that the range of ``yes'' and ``no'' may overlap. 
If a bundle satisfies PROP1 and its value is no more than the proportional value, the algorithm may output either ``yes'' or ``no''.
The correctness of this subroutine is stated in~\Cref{thm:oracle}.

\begin{restatable}{theorem}{PROPOnePROPOracle}
\label{thm:oracle} 
Given a bundle $B\subseteq M$ and a valuation function $u$, $B$ is PROP1 if \PropOracle$(n, M, B, u)$ outputs ``yes'' and is not PROP if it outputs ``no''.
In addition, \PropOracle always terminates in $O\left(n\log m\right)$ queries with the same running time.            
 \end{restatable}
\begin{proof}
In each iteration, the query complexity of finding $D_i$ and  $I_i$ is $O(\log m)$ using a binary search.
Therefore, the overall query complexity and running time of Algorithm~\ref{alg:oracle} are both $O(n\log m)$.
Next, we demonstrate the correctness of the output by the subroutine. 

Without loss of generality, assume $u(M)=1$.
Suppose the subroutine terminates after $\ell$ rounds, where $\ell \le n-1$.
If the algorithm outputs ``yes'', we prove that $B$ is PROP1 by contradiction.
Let $g$ be the most valuable item among all the items contained in $I_1,\ldots,I_\ell$.
If $B$ is not PROP1, then $u(B\cup\{g\})<1/n$.
Hence,
\begin{align*}
u(M) & = u(B) + \sum_{i=1}^\ell u(D_i\cup I_i) + u(R) = u(B) + \sum_{i=1}^\ell u(D_i\cup I_i) \tag{as ``yes'' is returned} \\
& \le u(B\cup \{g\}) + \sum_{i=1}^\ell u(D_i\cup I_i) \le (\ell+1)\cdot u(B\cup \{g\}) < \frac{\ell+1}{n} \le 1,\tag{as $\ell\le n-1$} 
\end{align*}
which contradicts the assumption that $u(M)=1$.
Therefore, $B$ satisfies PROP1.
In addition, if the algorithm outputs ``no'', we demonstrate that $B$ is non-PROP.
Observe that
$$
u(M)=u(B)+\sum_{i=1}^{n-1}u(D_i\cup I_i)+u(R)>\sum\limits_{i=1}^{n}u(B)+u(R)=n{\cdot}u(B)+u(R)\,.
$$
This implies that $u(B)<(1-u(R))/n \le 1/n$, which indicates that $B$ is non-PROP.
\end{proof}

\subsection{PROP1 under Identical Valuation}
\label{sec:prop1_identical_val}

\begin{algorithm}[t]
\caption{Computing a PROP1 allocation under identical valuation}\label{alg:prop1_identical}
\SetKwProg{Fn}{Function}{:}{}
{\tt \color{myPink} \small $\triangleright$ Phase I: Run subroutine $\ItemPartition([n], M)$}

\Fn{$\ItemPartition([n], M)$\label{alg:itempartition}}
{
Let $P\leftarrow M$ be the set of unallocated items\;
\While{$|P|\ge n$}{ \label{alg:prop1_fst_round_start}
    Partition $P$ evenly into $n$ bundles $(X_1, \dots, X_n)$ such that $\big| |X_i|-|X_j| \big| \le 1$ for all $i,j$\;
    Let $B_k\leftarrow A_k\cup X_k$ for each $k\in [n]$\;
    Let $i\leftarrow \argmin_{k\in[n]} u(B_k)$ be the index of the bundle among $(B_1,\dots,B_n)$ with the smallest value\;
    Update $A_i\leftarrow B_i$, $P\leftarrow P\setminus X_i$\;\label{alg:prop1_fst_round_end}
}
\Return{$(A_1, \ldots, A_i), P$}\;
}

$(A_1, \ldots,A_n), P\leftarrow \ItemPartition([n], M)$\;

{\tt \color{myPink} \small $\triangleright$ Phase II: Allocate the remaining items in $P$}

Sort $u(A_1), \ldots, u(A_n)$ and without loss of generality, we assume $u(A_1)\le  \dots\le u(A_n)$\label{alg:prop1_snd_round_start}\;
\While{$|P|> 0$}{
    Without loss of generality, assume $P=\{g_1,\dots, g_r\}$ and let $i\leftarrow \argmin_{k\in [r]} u(A_k\cup \{g_k\})$\label{alg:prop1_define_i}\;
    \If{$u(A_i\cup \{g_i\})\ge u(A_{r+1})$}{\label{alg:prop1_snd_round_if}
        Update $A_k\leftarrow A_k\cup \{g_k\}$ for $k\in[r]$, $P\leftarrow \emptyset$ \label{line:update_A_k}\;
        \textbf{break}\;
    }
    Update $A_i\leftarrow A_i\cup \{g_i\}$, $P\leftarrow P\setminus \{g_i\}$\;\label{alg:prop1_snd_round_end}

    Swap bundles to maintain the sorting $u(A_1) \le ... \le u(A_n)$;
}
\Return{the allocation $\mathcal{A}=(A_1, A_2,\cdots, A_n)$}
\end{algorithm}

Algorithm~\ref{alg:prop1_identical} considers a special case where the valuations are identical, represented by $u(\cdot)$.
The first phase of Algorithm~\ref{alg:prop1_identical} introduces a subroutine called \ItemPartition.
It begins with an empty allocation $\mathcal{A} = (\emptyset, \ldots, \emptyset)$.
Then, in each round, it evenly partitions the unallocated items into $n$ bundles $X_1,\ldots, X_n$ and attempts to append $X_i$ to agent $i$'s bundle $A_i$.
For the agent holding the bundle with the minimal utility, her bundle is updated to $A_i\cup X_i$.
For the other agents, their bundles remain unchanged, and $X_i$ is returned to the pool of unallocated items.
When fewer than $n$ items remain, \ItemPartition may no longer update the allocation: there may exist an agent $i$ such that $A_i\cup X_i$ has the smallest value with $X_i=\emptyset$.

When fewer than $n$ items are unallocated, the second phase (Lines~\ref{alg:prop1_snd_round_start}-\ref{alg:prop1_snd_round_end}) is conducted.
By sorting all the bundles, we assume $u(A_1) \le \cdots\le u(A_n)$ without loss of generality.
Next, in each round, let $r$ be the number of remaining unallocated items $g_1, \dots, g_r$.
These items are then added to the $r$ bundles $A_1, \ldots, A_r$ with the smallest values. 
If the bundle with the smallest value remains among the $r$ bundles, the new item is assigned to this bundle, while the remaining items are returned to the pool. The bundles are then swapped to maintain the condition $u(A_1)\le \cdots \le u(A_n)$.
Otherwise, if the bundle with the smallest value is no longer among the $r$ bundles, we directly assign each item $g_k$ to bundle $A_k$ for each agent $k\in [r]$ and terminate our algorithm.
The correctness of Algorithm~\ref{alg:prop1_identical} is stated in \Cref{thm:prop1_identical}.
\begin{restatable}{theorem}{propOneidentical}
\label{thm:prop1_identical}
For identical and additive valuations, a PROP1 allocation can be found by Algorithm~\ref{alg:prop1_identical} through $O(n^2\log m)$ queries with running time $O(n^2\log m)$.
\end{restatable}

\begin{proof}
To prove the output allocation is PROP1, we first show a key property of the allocation found by the first phase: $u(A_i)\le 1/n$ holds for each $i \in [n]$.
Assume, for contradiction, that there exits an agent $i \in [n]$ such that $u(A_i) > 1/n$.
Consider the iteration of the while-loop in which $A_i$ was updated to its current bundle, denoted as $A_i = A_i' \cup X_i$. Here, $A_i'$ is the bundle before the update.
By construction, $B_i = A_i' \cup X_i$ is the smallest bundle among the $n$ bundles in the allocation $(B_1,\ldots,B_n)$.
Hence, $\sum_{j\in[n]} u(B_j)\ge n\cdot u(B_i)>1$, which contradicts to the fact that $u(M) =1$. Thus, we have $u(A_i)\le 1/n$ for each $i \in [n]$.

Next, we prove that the output allocation is always PROP1 by analyzing whether the ``if''-branch (Line~\ref{alg:prop1_snd_round_if} of Algorithm~\ref{alg:prop1_identical}) in the second phase is executed.
If the ``if'' condition is never met within the second phase, then the bundle $A_n$ remains unchanged throughout the loop, as the utilities of the smallest $r$ bundles never exceed the utility of $A_{r+1}$, which is weakly less than $u(A_n)$.
From the first phase, we know that $u(A_n) \le 1/n$.
As a result, after the loop in the second phase, every bundle has a utility of at most $1/n$.
Since the sum of utilities across all bundles equals $u(M) = 1$, equality must hold for all bundles. This implies that the returned allocation is PROP.
Otherwise, we consider the case where the ``if''-branch is executed and the loop terminates after the ``if'' condition is satisfied.
Using a similar analysis as before, we can claim that, before executing the update at Line~\ref{line:update_A_k}, every bundle of $A_1, \ldots, A_n$ satisfies $u(A_j) \le u(A_n) \le 1/n$ for all $j \in [n]$ as $A_n$ is never changed.
Hence, there must exist some $k \in [r]$ such that $u(A_k \cup \{g_k\}) \ge 1/n$.
After the update at Line~\ref{line:update_A_k}, the bundle $A_{r+1}$ has the smallest value. Thus, it suffices to prove $A_{r+1}$ satisfying PROP1.
Observe that, $u(A_{r+1}) \ge u(A_k)$ before updating $A_k$.
Since $u(A_k \cup \{g_k\}) \ge 1/n$, it follows that $u(A_{r+1}) \ge 1/n - u(\{g_k\})$ and concludes the proof.

Finally, we analyze the query complexity and time complexity.
In the first phase, we execute the \ItemPartition subroutine.
Finding the smallest bundle costs $n$ queries in each iteration and the size of the unallocated pool decreases by a factor of $1/n$ each time.
Let $T_1(m)$ be the running time of the first phase with $m$ items unallocated.
Hence, $T_1(m)=T_1(\frac{n-1}{n}{\cdot} m) + O(n) = O(n\log_{\frac{n}{n-1}} m)=O(n^2\log m)$.
In the second phase, there are less than $n$ remaining items, so it runs for at most $n$ rounds. 
In each iteration, we find the smallest bundle and swap bundles to update the sorting, which costs $O(n)$.
Hence, the overall query complexity and running time are both $O(n^2\log m + n^2)=O(n^2\log m)$.
\end{proof}

\subsection{PROP1 under Non-identical Valuation}
\label{sec:prop1_nonidentical_val}

Based on \PropOracle and Theorem~\ref{thm:prop1_identical}, we are ready to present Algorithm~\ref{alg:prop1_general}, which can find a PROP1 allocation in $O(n^4 \log m)$ queries for additive valuations. 

\begin{algorithm}[t]
\caption{Computing PROP1 allocations for non-identical additive valuations}\label{alg:prop1_general}
Let $A_i\leftarrow \emptyset$ for $i\in[n]$\;
Let $P\leftarrow M$ be the set of unallocated items and $Q\leftarrow [n]$ be the set of agents not receiving items\;
\While{$|Q| > 0$}{
    Let an agent $k\in Q$ decide a PROP1 partition $\mathcal{B}= (B_1, \dots, B_{|Q|})$ of $P$ according to $u_k$ by Algorithm~\ref{alg:prop1_identical}\label{alg:additive_prop1_find_prop}\;
    Construct a bipartite graph $(Q, \mathcal{B}, E)$, where \label{alg:additive_prop1_construct_bip} 
    $$
    (i, B_j)\in E \quad \text{if $\PropOracle(|Q|,P, B_j, u_i) = yes$ or $i=k$}
    ;$$\\
    Let $Z\subseteq Q$ be the maximal subset such that $|Z|> |\Gamma(Z)|$, where $\Gamma(Z)$ denotes all the neighbors of $Z$\;
    Find a perfect matching between $Q\setminus Z$ and a subset of $\mathcal{B}\setminus \Gamma(Z)$\ \label{alg:prop1_general_matching}\;
    \ForEach{agent $i\in Q\setminus Z$ that is matched to $B_j\in \mathcal{B}\setminus \Gamma(Z)$}{
        Update $A_i\leftarrow B_j$, $Q\leftarrow Q\setminus \{i\}$ and $P\leftarrow P\setminus B_j$\;
    }
}
\Return{the allocation $\mathcal{A}=(A_1, A_2,\dots, A_n)$}
\end{algorithm}
\begin{figure}[t]
\centering
\begin{tikzpicture}[
    roundnode/.style={circle, draw=green!60, fill=green!5, very thick, minimum size=6mm},
    squarednode/.style={rectangle, draw=blue!60, fill=blue!5, very thick, minimum size=6mm},
    ]
    \node[roundnode] (L1) at (1.5, 0) {$1$};
    \node[roundnode] (L2) at (3, 0) {$2$};
    \node (dots1) at (4.5, 0) {$\cdots$};
    \node[roundnode] (Lk) at (6, 0) {$k$};
    \node (dots) at (7.5, 0) {$\cdots$};
    \node[roundnode] (Ln) at (9, 0) {$n$};

    \node[squarednode] (B1) at (1.5, -2) {$B_1$};
    \node[squarednode] (B2) at (3, -2) {$B_2$};
    \node (dots) at (4.5, -2) {$\cdots$};
    \node[squarednode] (Bk) at (6, -2) {$B_k$};
    \node (dots) at (7.5, -2) {$\cdots$};
    \node[squarednode] (Bn) at (9, -2) {$B_n$};

    \draw[semithick] (L1) -- (B1);
    \draw[semithick] (L2) -- (B1);
    \draw[semithick] (Lk) -- (B1);
    \draw[semithick] (Lk) -- (B2);
    \draw[preaction={draw,red!30,line width=5pt}] (Lk) -- (Bk);
    \draw[semithick] (Lk) -- (Bn);
    \draw[semithick] (Lk) -- (Bn);
    \draw[semithick] (Ln) -- (Bk);
    \draw[preaction={draw,red!30,line width=5pt}] (Ln) -- (Bn);

    \draw[thick, dashed, rounded corners] (1, 0.5) rectangle (3.5, -0.5);
    \draw[thick, dashed, rounded corners] (1, -1.5) rectangle (2, -2.5);
    \draw[thick, dashed, rounded corners] (3.6, 0.5) rectangle (9.5, -0.5);
    \draw[thick, dashed, rounded corners] (2.1, -1.5) rectangle (9.5, -2.5);
            
    \node at (0.5, 0) {$Z$};
    \node at (0.5, -2) {$\Gamma(Z)$};
    \node at (10.2, 0) {$\bar{Z}$};
    \node at (10.2, -2) {$\bar{\Gamma}(Z)$};
    \node at (12, 0) {\bf Agents};
    \node at (12, -2) {\bf Bundles};
\end{tikzpicture}
\caption{Illustration of the bipartite graph, where circle nodes and square nodes respectively represent the agents and bundles, and arrows in red shadow are the perfect matching.}
\label{fig:construct_of_bipa}
\end{figure}

\paragraph{Algorithm~\ref{alg:prop1_general}: overview.}
We begin with a high-level overview of Algorithm~\ref{alg:prop1_general}.
We first choose an arbitrary agent $k$ and let her decide a PROP1 allocation according to her valuation $u_k$ using Algorithm~\ref{alg:prop1_identical}.
Let the resulting allocation be $\mathcal{B} = (B_1,\dots, B_n)$, where each bundle is PROP1 to agent $k$.
At Line~\ref{alg:additive_prop1_construct_bip}, a bipartite graph is constructed, where the two sets of vertices correspond to the $n$ agents and the $n$ bundles in $\mathcal{B}$.
An edge is added between agent $i$ and bundle $B_j$ if \PropOracle determines that $B_j$ satisfies PROP1 for agent $i$ (i.e., \PropOracle outputs ``yes''). Additionally, edges are added between agent $k$ and all bundles.
Since each bundle in $\mathcal{B}$ has value at least $1/n$ for agent $k$, no agent will answer ``no'' to all bundles.
Thus, the graph contains at least one edge.

Next, as illustrated in \Cref{fig:construct_of_bipa}, we greedily find a maximal set $Z$ of agents such that $|Z| > |\Gamma(Z)|$, where $\Gamma(Z)$ denotes the set of neighbor bundles of $Z$. Let $\bar{Z}$ and $\bar{\Gamma}(Z)$ denote the complements of $Z$ and $\Gamma(Z)$, respectively. Since agent $k$ is adjacent to all bundles (by construction) from Line~\ref{alg:additive_prop1_construct_bip}, $k\notin Z$.
Thus, $\bar{Z}$ is non-empty, i.e., $|\bar{Z}|>0$.
Using Hall's marriage theorem~\cite{hall1987representatives}, we show in Lemma~\ref{lem:prop1_hall_marriage} that a perfect matching always exists between $\bar{Z}$ and a subset of $\bar{\Gamma}(Z) = \mathcal{B}\setminus \Gamma(Z)$.
In this perfect matching, each agent in $\bar{Z}$ is matched to an adjacent bundle in $\bar{\Gamma}(Z)$, while some bundles in $\bar{\Gamma}(Z)$ may remain unmatched. Each matched bundle is then allocated to the corresponding matched agent.
For the remaining agents in $Z$ and the unallocated items, we repeat the above process with fewer agents and items until each agent receives a bundle.
Below, we prove that the allocation output by Algorithm~\ref{alg:additive_prop1_find_prop} is always PROP1 and the algorithm terminates in $O(n^4\log m)$ queries and time.

\begin{restatable}{theorem}{PropOneGeneral}
\label{thm:prop1_general}
For non-identical additive valuations, a PROP1 allocation can be found by Algorithm~\ref{alg:prop1_general} via $O(n^4\log m)$ queries with running time $O(n^4\log m)$.
\end{restatable}
\begin{proof}
We first show that a perfect matching always exists between $\bar{Z}$ and a subset of $\bar{\Gamma}(Z)$ at Line~\ref{alg:prop1_general_matching}, which guarantees the validity of the algorithm.
\begin{restatable}{lemma}{PropOneHallMarriage}
\label{lem:prop1_hall_marriage}
A perfect matching always exists between $\bar{Z}= Q\setminus Z$ and a subset of $\bar{\Gamma}(Z) = \mathcal{B}\setminus \Gamma(Z)$ in Line~\ref{alg:prop1_general_matching} of Algorithm~\ref{alg:prop1_general}.
\end{restatable}
\begin{proof}
We show that the induced subgraph containing $\bar{Z}$, $\bar{\Gamma}(Z)$, and the edges between them satisfoes Hall's marriage theorem~\cite{hall1987representatives}.\footnote{Hall's marriage theorem states that for a bipartite graph $G=(A, B)$, if for any subset $S\subseteq A$, we have $|\Gamma(S)|\ge |S|$, then there exists a perfect matching such that each vertex in $A$ is matched to an adjacent vertex in $B$.}
Each vertex $v\in \bar{Z}$ is adjacent to at least one vertex in $\bar{\Gamma}(Z)$. For each subset $W\subseteq \bar{Z}$, $|W| \le |\Gamma(W)|$.
Otherwise, in both cases, we may add $v$ or $W$ to $Z$ to obtain a larger $Z$, which contradicts the maximality of $Z$.
Hence, a perfect matching always exists between $\bar{Z}= Q\setminus Z$ and a subset of $\bar{\Gamma}(Z) = \mathcal{B}\setminus \Gamma(Z)$.
\end{proof}

Next, we prove that the allocation output by Algorithm~\ref{alg:prop1_general} is always PROP1.
The high-level idea is that if a bundle $B_j$ is allocated to one agent $i\in Q$, then $B_j$ is PROP1 to agent $i$.
For any other agent $k\in Q$ not receiving any bundle, $u_k(B_j)$ is less than the proportional share of $u_i(P)$ since there exists no edge between $k$ and $B_j$.
Thus, she still can get a PROP1 bundle in the following allocation.

First, we prove that each agent will receive a bundle when Algorithm~\ref{alg:prop1_general} terminates.
In each iteration of the ``while'' loop, each matched agent in set $\bar{Z}$ receives a bundle and is subsequently removed from the agent set $Q$.
Since $|\bar{Z}|>0$, at least one agent is matched to a bundle in each iteration.
Therefore, every agent is guaranteed to receive a bundle.

Then, we prove by induction that each agent's allocated bundle satisfies PROP1.
For agents in $\bar{Z}$ who receives bundles from the matching in the first iteration, their allocated bundles are PROP1 since \PropOracle outputs ``yes''.
For agents in $Z$ that do not receive any bundles, the value of each matched bundle in $\bar{\Gamma}(Z)$ is less than their proportional value over the whole item set. This is because \PropOracle outputs ``no'', and such bundles belong to $\Gamma(Z)$ and, therefore, cannot be matched.
We then consider an agent $i$ who is matched in the $k$-th iteration ($k>1$).
Denote by $n_k$ the number of unmatched agents before the $k$-th iteration and denote by $V_k$ the total value of unallocated items before the $k$-th iteration under valuation function $u_i$.
Similar to the above analysis, agent $i$ receives a PROP1 bundle $A_i$ with respect to $n_k$ and $V_k$, and she thinks each bundle allocated in $(k-1)$-th iteration is less than $V_{k-1}/n_{k-1}$.
This leads to that there exists $g$ such that
$$u_i(A_i\cup\{g\})\ge \frac{V_k}{n_k}> \frac{1}{n_k}\cdot \left(V_{k-1}-(n_{k-1}-n_k)\cdot \frac{V_{k-1}}{n_{k-1}}\right) = \frac{V_{k-1}}{n_{k-1}},$$
which indicates $A_i$ is also PROP1 to agent $i$ in the $(k-1)$-th iteration.
Repeating this process, we can conclude that $u_i(A_i\cup\{g\})\ge 1/n$, thereby completing the induction step.

Finally, we analyze the query complexity and time complexity of Algorithm~\ref{alg:prop1_general}.
As we have shown in \Cref{thm:prop1_identical}, the complexity for computing a PROP1 allocation at Line~\ref{alg:additive_prop1_find_prop} is $O(n^2\log m)$, and constructing a bipartite graph invokes \PropOracle $n(n-1)$ times between each agent and each bundle at Line~\ref{alg:additive_prop1_construct_bip}.
In addition, finding the maximal subset $Z$ costs $O(n^2)$, and if $|Z|=k$, finding the perfect matching (using the Hungarian algorithm) costs $O((n-k)^3)$. Neither process invokes any queries.
Denote the time and query complexity of Algorithm~\ref{alg:prop1_general} on a given fair instance by $T(n, m)$ and $Q(n, m)$.
By the above statement, we have
$$T(n, m )=T(k, m)+O\left(n(n-1) \cdot n\log m\right)+O(n^2)+O((n-k)^3).$$
In the worst case, $k=n-1$.
As $T(1, m) = O(1)$, we have $T(n, m)=O(n^4\log m)$.
Meanwhile, as $Q(1, m) = O(1)$, the query complexity can be obtained by
$$Q(n, m)=Q(k, m)+O\left((n-1)^2{\cdot} \left(n\log m \right)\right)=O\left(n^4\log m\right).\qedhere
$$
\end{proof}

\section{Query Complexity of EF1}\label{sec:ef1}
This section studies the query complexity of computing an EF1 allocation.
For a constant number of agents, the well-known Round-Robin algorithm that guarantees EF1 requires $O(m\log m)$ queries as each agent needs to sort the items in descending order, and the envy-graph procedure requires $O(m)$ queries as each round of the procedure allocates one item and costs a constant number of queries.
There is still a gap between the existing algorithms and the lower bound $\Omega(\log m)$ to compute an EF1 allocation under the comparison-based query model for a constant number of agents.
In this section, we endeavor to bridge this gap and improve the query complexity to compute an EF1 allocation.

When there are two agents, the algorithm \emph{cut-and-choose} outputs an EF1 allocation. 
The query complexity mainly depends on the first agent's division step, which can be implemented by a line binary search with $O(\log m)$ queries, and the second agent's choosing step only requires one query.
Thus, we have the following theorem and the detailed proof is deferred to~\Cref{appendix:proof_of_ef1_for_two_agents}.

\begin{restatable}{lemma}{TwoAgentsAddEFOne}
\label{thm:2agents_additive_ef1}
For two agents with additive valuations, the query complexity and running time of computing an EF1 allocation are $O(\log m)$.    
\end{restatable}

Next, we introduce our main result for this section, which shows that the query complexity for EF1 is $O(n^2\log m)$ when the valuations are identical.
When the number of agents is constant, it then provides a tight $\Theta(\log m)$ bound.

\begin{restatable}{theorem}{EFIdentical}
\label{thm:ef1_identical}
    An EF1 allocation can be found through $O(n^2\log m)$ queries with running time $O(n^2\log m)$ for identical and additive valuation functions.
\end{restatable}

The case of $m < n$ is straightforward.
Our algorithm focuses on $m \ge n$.
A bundle is said to be larger (or smaller) than another if its value with respect to $u(\cdot)$ is larger (or smaller) than the other's.

\medskip
\noindent {\it{Step 1:}} First run the subroutine \ItemPartition(Algorithm~\ref{alg:prop1_identical}), which returns a partial allocation $(A_1, \dots, A_n)$ and a set of unallocated items $P=\{g_1,\dots, g_{n-1}\}$.\footnote{We can verify that there will always be exactly $n-1$ unallocated items when \ItemPartition terminates.
When $m > n$, there will be a round where $n < |P| \le 2n$. In this situation, at most two items are allocated each round until $|P| \le n$. Therefore, there must exist a point where $|P| = n$ or $n-1$. If the former case occurs, it will then turn to $|P| = n-1$ after one more round.
}
Without loss of generality, assume $u(g_1)\le \dots \le u(g_{n-1})$ and $u(A_1)\le  \dots \le u(A_n)$.

\medskip
\noindent {\it Step 2:} We step back to the round where $A_n$ was last updated in \ItemPartition.
In that round, we have a temporary allocation $\mathcal{B}=(B_1, \dots, B_{n-1}, A_n)$ such that $A_n$ is the one with the smallest utility among all.
Let $B_{n-1}$ be the bundle containing $g_{n-1}$ (if not, switch it to $B_{n-1}$).
Arrange the bundles from the left to the right in the order of $B_1, \dots, B_{n-1}, A_n$ (as shown in \Cref{fig:step2}) and place $g_{n-1}$ in the right of $B_{n-1}$ such that $g_{n-1}$ is adjacent to $A_n$.

\begin{figure}[h]
    \centering
    \begin{tikzpicture}
        \draw[black, semithick] (0,0) -- (15,0);
        
        \foreach \x in {0, 3, 6, 9, 12, 15} {
            \fill[black] (\x,0) circle (2pt);
        }
        \node[below] at (1.5,-0.1) {$B_1$};
        \node[below] at (4.5,-0.1) {$B_2$};
        \node[below] at (6,-0.2) {$\cdots$};
        \node[below] at (7.5,-0.1) {$B_{n-2}$};
        \node[below] at (10.5,-0.1) {$B_{n-1}$};
        \node[below] at (13.5,-0.1) {$A_n$};
        
        \node[fill=blue!30, draw, minimum size=5pt] at (11.8,0) {};
        \node[above, blue] at (12,0.1) {\scriptsize $g_{n-1}$};
    \end{tikzpicture}
    \caption{Arrange the bundles in a line and place $g_{n-1}$ to the rightmost of $B_{n-1}$}
    \label{fig:step2}
\end{figure}

\medskip
\noindent {\it Step 3:}
Update the first $n-1$ bundles in a moving-knife manner.
As shown in \Cref{fig:step3}, we find $n-1$ new bundles $C_1, C_2, \ldots, C_{n-1}$ from left to right such that $C_i$ for any $i\in [n-1]$ is weakly larger than $A_n$ and smaller than $A_n$ after removing $C_i$'s rightmost item, which ensures the EF1 conditions from $C_i$ to $A_n$ and from $A_n$ to $C_i$.
As $A_n$ is not larger than $B_i$ with $i\in [n-1]$, the right endpoint of bundle $C_i$ is to the left of (or at most the same as) $B_i$, which means such an allocation (maybe partial) $\mathcal{C}=(C_1,\dots, C_{n-1}, A_n)$ exists.
If there is no unallocated item between $C_{n-1}$ and $A_n$, an EF1 allocation is already found.
\begin{figure}[h]
\centering
    \begin{tikzpicture}
            \draw[black, semithick] (0,0) -- (15,0);  
            % \draw[darkred, semithick] (0,0) -- (8,0);  
            % \draw[darkred, semithick] (8,0) -- (12,0);  
            \foreach \x in {0, 3, 6, 9, 12, 15} {
                \fill[gray!70] (\x,0) circle (2pt);
            }
            \foreach \x in {0, 2, 4, 6, 8} {
                \fill[black] (\x-0.05,0.1) -- (\x+0.05,0.1) -- (\x+0.05,-0.1) -- (\x-0.05,-0.1);
            }
            \fill[black] (12-0.05,0.1) -- (12+0.05,0.1) -- (12+0.05,-0.1) -- (12-0.05,-0.1);
            \foreach \x in {15} {
                \fill[black] (\x-0.05,0.1) -- (\x+0.05,0.1) -- (\x+0.05,-0.1) -- (\x-0.05,-0.1);
            }

            \node[above] at (1.5,0.1) {\color{gray} $B_1$};
            \node[above] at (4.5,0.1) {\color{gray}  $B_2$};
            \node[above] at (6,0.2) {\color{gray} $\cdots$};
            \node[above] at (7.5,0.1) {\color{gray}  $B_{n-2}$};
            \node[above] at (10.5,0.1) {\color{gray}  $B_{n-1}$};
            
            \node[below] at (1,-0.1) {$C_1$};
            \node[below] at (3,-0.1) {$C_2$};
            \node[below] at (4,-0.2) {$\cdots$};
            \node[below] at (5,-0.1) {$C_{n-2}$};
            \node[below] at (7,-0.1) {$C_{n-1}$};
            \node[below] at (10,-0.1) {$R$};
            \node[below] at (13.5,-0.1) {$A_n$};
            
            \node[fill=blue!30, draw, minimum size=5pt] at (11.8,0) {};
            \node[above, blue] at (12,0.1) {\scriptsize $g_{n-1}$};
        \end{tikzpicture}
        \caption{Run moving-knife and the $n-1$ new bundles $C_1, \ldots, C_{n-1}$}
        \label{fig:step3}
\end{figure}

\medskip
\noindent {\it Step 4:}
If there are still unallocated items between $C_{n-1}$ and $A_n$, denote them by $R$.
We have $g_{n-1}\in R$ as it is placed directly to the left of $A_n$ in Step 2.
If $g_{n-1}$ is the only item in $R$, we proceed to Step~5.
Otherwise, we proceed with the following operation, iterating $g_k$ in the order $g_{n-2}$, $g_{n-3}$, $\ldots$, $g_1$: Suppose $g_k$ is included in $C_i$.
Then we exchange $g_k$ with a set of items $I$ of $R$ by including $g_k$ into $R$ and moving $I$ into $C_i$ so that the updated $C_i$ is still weakly larger than  $A_n$ (to guarantee EF1). 
We illustrate the process in \Cref{fig:step4_1} and \Cref{fig:step4_2}.

Note that we prioritize using the items in $R\setminus P$ to construct $I$.
$I$ may intersect $P$, in which case the algorithm terminates this step with a partial allocation $\mathcal{D} = (D_1, \dots, D_{n-1}, A_n)$ and $S=\emptyset$.
Otherwise, we repeat this process until no more exchanges can be made or $R\subseteq P$.
Then, we obtain a partial allocation $\mathcal{D}$ and a set of unallocated items $R$.
Set $S= R\setminus P$.

\begin{figure}[h]
\centering
\subfloat[Step~4 before exchanging item $g_k$ with $I$]{
        \begin{tikzpicture}
            \draw[black, semithick] (0,0) -- (15,0);  

            \foreach \x in {0, 2, 4, 6, 8, 12, 15} {
                \fill[black] (\x-0.05,0.1) -- (\x+0.05,0.1) -- (\x+0.05,-0.1) -- (\x-0.05,-0.1);
            }
            
            \node[below] at (3,-0.1) {$C_i$};
            \node[below] at (10,-0.1) {$R$};
            \node[below] at (13.5,-0.1) {$A_n$};
            
            \node[fill=red!30, draw, minimum size=5pt] at (8.2,0) {};
            \node[fill=red!30, draw, minimum size=5pt] at (8.5,0) {};
            \node[fill=red!30, draw, minimum size=5pt] at (8.8,0) {};
            \node[above] at (8.5,0.1) {$I$};

            \node[fill=darkgreen!30, draw, minimum size=5pt] at (3,0) {};
            \node[above, darkgreen] at (3,0.1) {\footnotesize $g_k$};

            \node[fill=blue!30, draw, minimum size=5pt] at (10,0) {};
            \node[above, blue] at (10,0.1) {\footnotesize $g_{k+1}$};

            \node at (10.9,0.1) { $\cdots$};

            \node[fill=blue!30, draw, minimum size=8pt] at (11.8,0) {};
            \node[above, blue] at (12,0.1) {\scriptsize $g_{n-1}$};
        \end{tikzpicture}
        \label{fig:step4_1}
    }

    \subfloat[Step~4 after exchanging item $g_k$ with $I$]{
        \begin{tikzpicture}
            \draw[black, semithick] (0,0) -- (15,0);  

            \foreach \x in {0, 2, 4.6, 6.6, 8.6, 12, 15} {
                \fill[black] (\x-0.05,0.1) -- (\x+0.05,0.1) -- (\x+0.05,-0.1) -- (\x-0.05,-0.1);
            }
            
            \node[below] at (3.6,-0.1) {$D_i$};
            \node[below] at (10.6,-0.1) {$R$};
            \node[below] at (13.5,-0.1) {$A_n$};
            
            \node[fill=red!30, draw, minimum size=5pt] at (3,0) {};
            \node[fill=red!30, draw, minimum size=5pt] at (3.3,0) {};
            \node[fill=red!30, draw, minimum size=5pt] at (3.6,0) {};
            \node[above] at (3.3,0.1) {$I$};

            \node[fill=darkgreen!30, draw, minimum size=5pt] at (9.7,0) {};
            \node[above, darkgreen] at (9.6,0.1) {\footnotesize $g_k$};

            \node[fill=blue!30, draw, minimum size=5pt] at (10,0) {};
            \node[above, blue] at (10.1,0.1) {\footnotesize $g_{k+1}$};
            
            \node at (10.9,0.1) { $\cdots$};

            \node[fill=blue!30, draw, minimum size=8pt] at (11.8,0) {};
            \node[above, blue] at (12,0.1) {\scriptsize $g_{n-1}$};
        \end{tikzpicture}
        \label{fig:step4_2}
    }
\caption{Illustration of Step 4}
\end{figure}
\medskip
\noindent {\it Step 5:} Add $S$ to an arbitrary agent's bundle. 
For the remaining items in $R$, run the envy-graph procedure to allocate them based on the partial allocation $\mathcal{D}$.

\begin{proof}[Proof of \Cref{thm:ef1_identical}.]
% The proof directly follows from~\Cref{lem:ef1_identical_ef1} and~\Cref{lem:ef1_identical_exchange}.
We first prove that the allocation returned by the above procedure is EF1. 
As guaranteed in our algorithm, for each bundle $C_i$ in $\mathcal{C}$ after Step~3, $A_n$ is EF1 to $C_i$ and $C_i$ is EF to $A_n$. 
Then, for any pair of bundles $C_i$ and $C_j$, there exists an item $g\in C_j$ such that $u(C_i)\ge u(A_n)\ge u(C_j\setminus\{g\})$.
Hence, the (partial) allocation $\mathcal{C}$ satisfies EF1.
If a complete allocation has not been found, our algorithm proceeds to Step~4.
In Step~4, as we still ensure $A_n$ is EF1 to each updated bundle $D_k$, we can show that $\mathcal{D}$ also satisfies EF1 using a similar analysis as above.

Next, in Step 5, we show that the utility of $S$ will always be zero, which does not destroy EF1 and concludes the property of EF1.
It is clear that $u(S) = 0$ when $S = \emptyset$.
Next, we assume $S = R\setminus P \neq \emptyset$.
As we have shown in the proof of \Cref{thm:prop1_identical}, for each $A_i, i\in [n]$, we have $u(A_i)\le 1/n$.
Let $u(A_n)=\alpha$. 
Since $u(A_1)\le\dots\le u(A_n)$, we have 
$$\sum_{i=1}^{n-1} u(g_i) = u(M)- \sum_{i=1}^n u(A_i) \ge 1-n\cdot u(A_n)=1-n\alpha.$$
% We abuse the symbol $S$ to denote $R\setminus P$ throughout the algorithm.
Consider the moment when we have just completed the exchange of item $g_{k+1}$. 
The items in $R$ are $(R\setminus P) \cup \{g_{k+1}, \dots, g_{n-1}\}$.
Now we are about to exchange $g_{k}$ with some items in $R$ and assume $g_k$ is in the bundle $C_i$.
There are three possible cases for the upcoming exchange step.
\begin{itemize}
\item 
If $R\setminus P =\emptyset$, then the algorithm terminates with $S = \emptyset$;

\item 
If $C_i$ is weakly larger than $A_n$ after receiving a subset of $S$, i.e., there exists a subset $S'\subseteq S$ and an item $g\in S'$, such that $u(D_i)=u(C_i\setminus \{g_k\} \cup S')\ge u(A_n)$ and $u(C_i\cup S'\setminus \{g_k, g\})< u(A_n)$, we exchange $g_k$ with $S'$.
The order of $R$ will be updated to $S\setminus S', g_k, g_{k+1},\dots, g_{n-1}$;
% \bt{Does this sequence end at $g_{n-1}$? If so, put the ending item.}

\item
Otherwise, $C_i$ is still smaller than $A_n$ even when receiving the entire set $S$.
Hence, we have $u(C_i\cup S\setminus \{g_k\})<u(A_n)$.
Since $u(C_i\cup S\setminus \{g_k\}) <u(A_n)$ and $u(g_{k+1})\ge u(g_k)$, we have $u(C_i\cup S\cup\{g_{k+1}\}\setminus\{g_k\})\ge u(C_i\cup S)\ge u(A_n)$.
Therefore, we exchange $g_k$ with $S\cup\{g_{k+1}\}$ and terminates the algorithm.
$D_i$ is updated to $C_i \cup S\cup \{g_{k+1}\}\setminus \{g_{k}\}$ and the order of $R$ is updated to $g_k, g_{k+2},g_{k+3},\ldots,g_{n-1}$.
\end{itemize}
The only scenario in which the algorithm terminates with a nonempty set $S\neq \emptyset$ occurs in the second case.
In that case, $I$ does not intersect with $P$ throughout the step, hence, $R$ contains all the items from $g_1$ to $g_{n-1}$ and $S\neq \emptyset$.
If $u(S)>0$, as $u(D_i)\ge u(A_n)$ for $i\in[n-1]$, it leads to
$$u(M)=\sum_{i=1}^{n-1} u(D_i)+u(A_n)+\sum_{i=1}^{n-1} u(g_i)+u(S)\ge n\alpha+ (1-n\alpha)+u(S)>1,$$
which contradicts to the assumption $u(M)=1$.
% \bt{Is it true that $S$ will always be the empty set after step 4? If so, why don't we just say ``we demonstrate that $S=\emptyset$'' instead of saying the utility of $S$ is zero (at the begining of the second paragraph)? If this is the case, I think we should also remark it when describing Step 5.}

Finally, it remains to show its query complexity and running time complexity.
In Step~1, the subroutine \ItemPartition requires $O(n^2\log m)$ queries as analyzed in~\Cref{thm:prop1_identical}.
In Step~3, as the items are arranged in a line, the moving-knife method for each bundle can be implemented using binary search to find the rightmost end of the bundle $C_i$, costing $O(n\log m)$ queries in total.
Similarly, the exchange operation for each bundle in Step~4 can also be done via binary search in $O(\log m)$ queries.
As proved before, after Step~4, $R\subseteq P$ or $u(R\setminus P) = 0$.
Hence, there are at most $n-1$ remaining items to be allocated during the envy-graph procedure.
Constructing the envy graph for the first time costs $O(n\log n)$ queries to sort the $n$ bundles.
Then, in each iteration, after an item is allocated to the smallest bundle, we only need to sort this specific bundle based on the original order, which can be done using $O(\log n)$ queries via binary search. 
Hence, we conclude that both the query complexity and running time are $O(n^2\log m)$.
\end{proof}

We further study the query complexity for three agents with additive valuations and show it is still $O(\log m)$.
This algorithm is almost a straightforward adaptation of~\cite[Algorithm 3]{Oh_2021}, and the above algorithm for identical valuation is also used.
The detailed algorithm and proof are deferred to~\Cref{sec:detailed_proof_of_ef1_3agents}.

\begin{restatable}{theorem}{EfThree}
\label{thm:3agents_additive_ef1}
For three agents with additive valuations, the query complexity and running time of computing an EF1 allocation are both $O(\log m)$. 
\end{restatable}

\section{Query Complexity of MMS+PROP1}
\label{sec:prop1+mms}
In the previous sections, we have already shown that a PROP1 allocation can be found via $O(\log m)$ queries for a constant number of agents with additive valuations.
In addition to PROP1, MMS is also a well-studied fairness notion in the fair division of indivisible goods.
This section provides an algorithm that can achieve PROP1 and $1/2$-MMS simultaneously via $O(n^4\log m)$ queries, which is shown in Algorithm~\ref{alg:prop1_mms_general}.

Before the detailed descriptions of our algorithm, we first define two types of bundles that are allocated to the agents during our algorithm. 
We call the bundle $A_k$ in Line~\ref{alg:prop1_mms_part1case1end} and the bundle $A_i$ in Line~\ref{alg:prop1_mms_part2case1end} \emph{large bundle} and the bundle $A_i$ in Line~\ref{alg:prop1_mms_part3end} \emph{normal bundle}.
A large bundle is a singleton containing a valuable item while a normal bundle may not consist of only one item.

\paragraph{Overview of Algorithm~\ref{alg:prop1_mms_general}.}
Similar to Algorithm~\ref{alg:prop1_general}, we perform a multi-round algorithm.
In each round, according to the set $Q$ (agents without receiving any items) and the set $P$ (unallocated items), we first select an arbitrary agent $k\in Q$ and find a PROP1 and $1/2$-MMS allocation subject to agent $k$'s valuation $u_k$, corresponding to the first phase of Algorithm~\ref{alg:prop1_mms_general} (Lines~\ref{alg:prop1_mms_part1begin}-\ref{alg:prop1_mms_part1end}).
In this process, we invoke subroutine $\ItemPartition(Q,P)$ defined in Algorithm~\ref{alg:prop1_identical} to find the partial allocation $\mathcal{B}$ and a set $P'$ of unallocated items whose size is less than the number of remaining agents in $Q$, i.e., $|Q|$.
If the two most valuable items of $P'$ are more valuable than any bundle in $\mathcal{B}$, then we can directly allocate the most valuable item in $P'$ to agent $k$ and restart the main procedure for the remaining items and agents by ignoring all previous normal bundles (corresponding to Lines~\ref{alg:prop1_mms_part1case1begin}-\ref{alg:prop1_mms_part1case1end}).
We remark that the singleton allocated satisfies PROP1+$1/2$-MMS.
Otherwise, we can run the procedure described in~\Cref{thm:ef1_identical} to obtain an EF1 allocation $\mathcal{C}$ under $u_k$, which is the target PROP1 and $1/2$-MMS allocation under valuation $u_k(\cdot)$.

\begin{algorithm}[htbp]
\caption{Computing PROP1+$1/2$-MMS allocation}\label{alg:prop1_mms_general}
\small
\SetKwProg{Fn}{Function}{:}{}
\Fn{$\MainAlgorithm(N, M)$\label{alg:prop1mms_main}}
{
\KwAssume{The agent set $N=\{1,\ldots,n\}$, the item set $M=\{g_1,\ldots,g_m\}$}
Let $A_i\leftarrow \emptyset$ for each $i\in N$, $P\leftarrow M$ be the set of unallocated items, and $Q\leftarrow N$ be the set of agents without receiving any item\;
\While{$|Q| > 0$}{

{\tt \color{myPink} \small $\triangleright$ Phase I: Find a PROP1 + $\frac12$-MMS allocation for an arbitrary agent $k$}

    Choose an arbitrary agent $k\in Q$\; \label{alg:prop1_mms_part1begin}
    Run $\ItemPartition(Q,P)$ under the valuation $u_k$ and obtain a partial allocation $\mathcal{B}=(B_1,\ldots,B_{|Q|})$ with the set of unallocated items $P'$\;\label{alg:prop1_mms_part1sub}
    Let $\mathcal{P}'=\{\{g\}:g\in P'\}$ be the set of all singletons\;
    \If{the two most valuable bundles among $\mathcal{B}\cup\mathcal{P}'$ under $u_k$ are both in $\mathcal{P}'$\label{alg:prop1_mms_part1case1begin}}
    {
    Set $A_k$ as the bundle which contains only the most valuable item in $P'$ under $u_k$\;
    Allocate bundle $A_k$ to agent $k$\;
    Let $\mathcal{A}'$ be the allocation returned by $\MainAlgorithm(Q\setminus\{k\},P\setminus A_k))$\;
    \Return{the allocation by merging the current allocation and $\mathcal{A}'$}\;\label{alg:prop1_mms_part1case1end}
    }
    Let $u'$ be the utility of the most valuable bundles among $\mathcal{B}$ under $u_k$\;\label{alg:prop1_mms_part1_defup}
    Execute Step~2 to Step~5 in the procedure described in Theorem~\ref{thm:ef1_identical} on the instance $(Q, P)$ and obtain an EF1 allocation $\mathcal{C}=(C_1,\ldots,C_{|Q|})$ under valuation $u_k$\; \label{alg:prop1_mms_part1end}

{\tt \color{myPink} \small $\triangleright$ Phase II: Construct bipartite graph $G$ where $(i, C_j)\in E$ if $C_j$ is PROP1 + $\frac12$-MMS to $i$}

    Construct a bipartite graph $G=(Q,\mathcal{C},E)$ where $E\leftarrow\emptyset$ is the initial edge set\;\label{alg:prop1_mms_part2begin}
    Create an undirected edge $(k,C_j)$ for every bundle $C_j\in \mathcal{C}$\;\label{alg:prop1_mms_part2edge}
    % Line up the items in $P$ in the order $(C_1,\ldots,C_{|Q|})$\;
    \ForEach{agent $i\in Q\setminus\{k\}$ and bundle $C_j\in \mathcal{C}$}
    {
    Execute $\PropOracle(|Q|, P, C_j, u_i)$ with lining up the items in $P$ in the order $(C_1,\ldots,C_{|Q|})$ and record the bundles after partitioning as $\{C_j,D_1,I_1,\ldots,D_{|Q|-1},I_{|Q|-1},R\}$\label{alg:prop1_mms_part2pairbegin}\label{alg:prop1_mms_part2construction}\;
    \If{``yes'' is returned}
    {   
        Let $\mathcal{D}\leftarrow \{C_j, D_1,\ldots,D_{|Q|-1}\}$ and $\mathcal{I} \leftarrow \{I_1,\ldots,I_{|Q|-1}\}$\;
        \If{the two most valuable bundles among $\mathcal{D}\cup \mathcal{I}$ under $u_i$ are both in $\mathcal{I}$\label{alg:prop1_mms_part2case1begin}}
        {
            Set $A_k$ as a singleton only containing the most valuable item in $\mathcal{I}$\;
            Allocate $A_k$ to agent $k$\;
            Let $\mathcal{A}'$ be the allocation returned by $\MainAlgorithm(Q\setminus\{k\},P\setminus A_k))$\;
            \Return{the allocation by merging the current allocation and $\mathcal{A}'$}\;\label{alg:prop1_mms_part2case1end}
        }
        \Else
        {
            Add an undirected edge $(i,C_j)$ on $G$\;\label{alg:prop1_mms_part2end}
        }
    }
    }
    
    {\tt \color{myPink} \small $\triangleright$ Phase III: Find a perfect matching and allocate items}
    
    Let $Z\subseteq Q$ be the maximal subset such that $|Z|> |\Gamma(Z)|$\; \label{alg:prop1_mms_part3begin}
    Find a perfect matching between $Q\setminus Z$ and a subset of $\mathcal{C}\setminus \Gamma(Z)$\;\label{alg:prop1_mms_part3perfect}
    \ForEach{agent $i\in Q\setminus Z$ that is matched to $C_j\in \mathcal{C}\setminus \Gamma(Z)$}{
        Update $A_i\leftarrow C_j$, $Q\leftarrow Q\setminus \{i\}$ and $P\leftarrow P\setminus C_j$\;\label{alg:prop1_mms_part3end}
    }
} 
\Return{the allocation $\mathcal{A}=(A_1, A_2,\dots, A_n)$}\;
}
\end{algorithm}

We then perform Lines~\ref{alg:prop1_mms_part2begin}-\ref{alg:prop1_mms_part2end} to construct a bipartite graph $G$ between all remaining agents in $Q$ and all bundles in the above EF1 allocation $\mathcal{C}$, similar to the bipartite construction in Algorithm~\ref{alg:prop1_general}. Here, an edge between an agent $i$ and a bundle $C_j$ is constructed only if allocating $C_j$ to agent $i$ can meet agent $i$'s requirement for both PROP1 and $1/2$-MMS.
We first execute the \PropOracle to identify the bundles that satisfy PROP1 to agent $i$.
However, different from that used in Algorithm~\ref{alg:prop1_general}, we also need to ensure $1/2$-MMS here. 
Thus, we consider the allocation in two subcases.
Lines~\ref{alg:prop1_mms_part2case1begin}-\ref{alg:prop1_mms_part2case1end} perform the first case where agent $i$ can receive only one item in $\{I_1,\ldots,I_{|Q|-1}\}$ to reach both PROP1 and $1/2$-MMS.
Under this case, we can again allocate the most valuable item to agent $i$ and restart the main procedure with only large bundles allocated.
Otherwise, we can claim that agent $i$ can meet the requirement for PROP1 and $1/2$-MMS after receiving $C_j$, corresponding to Line~\ref{alg:prop1_mms_part2end} of Algorithm~\ref{alg:prop1_mms_general}.

Lines~\ref{alg:prop1_mms_part3begin}-\ref{alg:prop1_mms_part3end} perform a similar reduction procedure as in Algorithm~\ref{alg:prop1_general}, where we find a maximal subset violating Hall's marriage theorem. We then find a perfect matching of the remaining agents to some bundles and reduce the original instance to a smaller one after removing these agents and their matched bundles.

We now present our main theorem: Algorithm~\ref{alg:prop1_mms_general} finds both a PROP1 and $1/2$-MMS allocation in $O(n^4\log m)$ queries.
Similar to the proof of Theorem~\ref{thm:prop1_general}, our high-level idea is to iteratively reduce the input instance's size and demonstrate the correctness of the recursion.

\begin{restatable}{theorem}{PropOneMMSGeneral}
\label{thm:prop1_mms_general}
For additive valuations, a PROP1 and $1/2$-MMS allocation can be found by Algorithm~\ref{alg:prop1_mms_general} in $O(n^4\log m)$ queries with running time $O(n^4\log m)$.
\end{restatable}
\begin{proof}
We prove the correctness of $\MainAlgorithm(N,M)$ in Algorithm~\ref{alg:prop1_mms_general} by induction based on the recursive calls of the algorithm.
As specified, there are two types of bundles: a large bundle is a singleton (Lines~\ref{alg:prop1_mms_part1case1end} and~\ref{alg:prop1_mms_part2case1end}) while a normal bundle may consist of multiple items (Line~\ref{alg:prop1_mms_part3perfect}).
In the base case, as no recursive call is invoked, no large bundle is allocated. 

\paragraph{Base case: no recursive call.}
To show the output allocation $\mathcal{A}$ satisfies PROP1+$1/2$-MMS, we adapt the analysis of \Cref{thm:prop1_general} here.
We first show that the allocated bundles are PROP1+$1/2$-MMS in the smaller instance with agent set $Q$ and item set $P$.
In \Cref{lem:prop1_mms_part1}, we prove that any bundle in the allocation $\mathcal{C}$ satisfies PROP1 and $1/2$-MMS for the fixed agent $k$.
Then \Cref{lem:prop1_mms_part2edge} validates the correctness of the construction of the bipartite graph: a bundle is PROP1 and $1/2$-MMS to an agent if they are adjacent.
After that, we extend our results to show the allocated bundles are PROP1+$1/2$-MMS on the larger instance with agent set $N$ and item set $M$ in \Cref{lem:prop1_mms_small} and \Cref{lem:prop1_mms_total2}.

\begin{restatable}{lemma}{PropMmsPartOne}
\label{lem:prop1_mms_part1}
In the fair division instance of allocating $P$ among agents $Q$, the allocation $\mathcal{C}$ output at Line~\ref{alg:prop1_mms_part1end} is PROP1 and $1/2$-MMS under valuation $u_k(\cdot)$.
\end{restatable}
\begin{proof}
From the procedure described in \Cref{thm:ef1_identical}, the allocation $\mathcal{C}$ is an EF1 allocation under the identical valuation $u_k(\cdot)$ and $u_k(C_i)\ge u'$ for each $C_i$, where $u'=\max_{i\in Q} u_k(B_i)$ is defined in Line~\ref{alg:prop1_mms_part1_defup} of Algorithm~\ref{alg:prop1_mms_general}.
Consider the partial allocation $\mathcal{B}$ and the unallocated items $P'$ output by $\ItemPartition(Q,P)$ under the valuation $u_k$, since the condition at Line~\ref{alg:prop1_mms_part1case1begin} is not met, $u'$ is one of the two largest values among the values of bundles in $\mathcal{B}$ and items in $P'$. 
If $u'$ is the largest one, we have 
$$u'\ge \frac{1}{|Q|+|P'|}u_k(P)\ge \frac{1}{2|Q|}u_k(P)\ge \frac12 \cdot \MMS$$ 
from $|P'|<|Q|$ and \Cref{lem:prop1_mms_total2}.
If $u'$ is the second largest one, the largest one should be a single item in $P'$ (denoted by $o$).
Hence, after removing this singleton, $u'$ is still the largest bundle.
Therefore,
$$
u'\ge \frac{1}{|P'|+|Q|-1}u_k(P\setminus\{o'\})
\ge \frac{1}{2}\frac{1}{|Q|-1}u_k(P\setminus\{o'\})
\ge \frac{1}{2}\cdot \MMS'\,.
$$
where $\MMS'$ is the maximin share after removing $o$ and an arbitrary agent $i\in Q$.
Assume $\MMS$ is obtained from an allocation $\mathcal{X}$.
We construct a new allocation for agents $Q\setminus \{i\}$ and items $P\setminus \{o\}$ by appending the other items of the bundle containing $o$ to another arbitrary bundle of $\mathcal{X}$, which weakly improves the utility of the least preferred bundle.
Therefore, $\MMS'\ge \MMS$ and $u'\ge 1/2\cdot \MMS$ always holds.
As $u_k(C_i) \ge u'$ for any $C_i$ and $\mathcal{C}$ is EF1, we can conclude $\mathcal{C}$ is both PROP1 and $1/2$-MMS under valuation $u_k(\cdot)$.
\end{proof}

\begin{restatable}{lemma}{PropMmsPartTwoEdge}
\label{lem:prop1_mms_part2edge}
In the instance of allocating $P$ among agents $Q$, bundle $C_j$ is PROP1 and $1/2$-MMS to an agent $i$ if they are adjacent in $G$.
Otherwise, bundle $C_j$ is not PROP to agent $i$.
\end{restatable}
\begin{proof}
For every agent $i$ and bundle $C_j$, if ``no'' is returned, by the property of \PropOracle, bundle $C_j$ is non-PROP to agent $i$.
Hence, it remains to prove that every added edge satisfies the requirement of PROP1 and $1/2$-MMS.    
The PROP1 condition is satisfied for the same reason as in~\Cref{thm:oracle}.
For MMS, since the condition at Line~\ref{alg:prop1_mms_part2case1begin} is not met, $C_j$ must be one of the two largest among $\mathcal{D}\cup \mathcal{I}$ under $u_i(\cdot)$.
From the same analysis as in the proof of \Cref{lem:prop1_mms_part1}, the $1/2$-MMS condition for agent $i$ receiving the bundle $C_j$ is satisfied.
\end{proof}

By \Cref{lem:prop1_mms_part1}, every bundle $C_j$ is PROP1+$1/2$-MMS to agent $k$, thereby agent $k$ is included in $\bar{Z}$ and each iteration the size of $Q$ reduces by at least one.
Therefore, the procedure must terminate and every agent in $N$ receives one bundle.
By \Cref{lem:prop1_mms_part2edge}, for every pair of adjacent of agent $i$ and $C_j$, bundle $C_j$ is PROP1+$1/2$-MMS to agent $i$ under the smaller instance $(Q, P)$.
Next, we demonstrate that bundle $C_j$ is also PROP1+$1/2$-MMS to agent $i$ under the larger instance $(N, M)$, which further concludes the proof of the induction basis.

We adapt a similar analysis to the proof of \Cref{thm:prop1_general}.
In \Cref{lem:prop1_mms_small}, we show that there always exists an allocation of the remaining items for an agent not receiving any item yet, such that every bundle of the allocation has a weakly higher utility than any allocated bundle.
It then implies the maximin share of this agent does not decrease as the instance is reduced.
Afterward, using \Cref{lem:prop1_mms_small}, we show the proportional value also does not decrease as the reduction of the instance.

\begin{lemma}\label{lem:prop1_mms_small}
In any round that allocates a set $P$ of items to a set $Q$ of agents, for each remaining agent $i\in Z$, we can partition the set of the remaining items into $|Z|$ bundles such that $i$ values each bundle weakly larger than those allocated bundles in this round at Line~\ref{alg:prop1_mms_part3end}.
\end{lemma}
\begin{proof}
    Fix an agent $i\in Z$ and we assume that the allocated bundles in the passing round are $(C_{i(1)},\ldots, C_{i(t)})$ where $t= |\bar{Z}|$, $u_i(C_{i(1)})\ge \ldots \geq u_i(C_{i(t)})$ and $\abs{Z}+t = \abs{Q}$.
    % We further assume a permutation $(i(1),\ldots,i(t))$ of $[t]$ such that $u_i(C_{i(1)})\ge \ldots \geq u_i(C_{i(t)})$.
    Arrange items in $P\setminus C_{i(1)}$ in the same order as their order in $C_1,\dots, C_{\abs{Q}}$ and then execute \PropOracle for agent $i$ and bundle $C_{i(1)}$.
    Suppose the recorded partition is $\{C_{i(1)},D_1,\ldots,D_{|Q|-1},I_1,\ldots,I_{|Q|-1},R\}$.
    We then construct a new allocation with $\abs{Z}$ bundles such that every one has weakly higher value than any of $C_{i(j)}$ with $j\in[t]$ using the partition.
    We start from bundle $C_{i(1)}$.
    As $i$ is not adjacent to $C_{i(1)}$, \PropOracle outputs ``no'' for agent $i$ and bundle $C_{i(1)}$.
    Hence, $u_i(D_\ell \cup I_\ell) \ge u_i(C_{i(1)})$ for any $\ell \le \abs{Q} - 1$.
    We then can treat each $I_\ell\cup D_\ell$ as one of the initial $|Q|-1$ bundles $(W_1,\ldots,W_{|Q|-1})$ of item set $P\setminus C_{i(1)}$ with $W_\ell = D_\ell \cup I_\ell$ for $\ell < \abs{Q} - 1$ and $W_{\abs{Q}-1} = D_{\abs{Q}-1} \cup I_{\abs{Q}-1} \cup R$.
    \begin{figure}[h]
        \centering
        \begin{tikzpicture}
            \filldraw[blue!30, draw=black, line width=1pt] (0-1,0.3) -- (0+1,0.3) -- (0+1,-0.3) -- (0-1,-0.3) -- cycle;
            \node[] at (0, 0) {$C_{i(1)}$};
            \fill[pattern=north east lines, pattern color=gray] (0-1,-1+0.3) -- (0+1,-1+0.3) -- (0+1,-1+-0.3) -- (0-1,-1+-0.3) -- cycle;
            \filldraw[blue!30, draw=black, line width=1pt] (2.1-1,0.3) -- (2.1+1,0.3) -- (2.1+1,-0.3) -- (2.1-1,-0.3) -- cycle;
            \node[] at (2.1, 0) {$C_{1}$};
            \node[] at (4.6, 0) {$\cdots$};
            \filldraw[blue!30, draw=black, line width=1pt] (7.1-1,0.3) -- (7.1+1,0.3) -- (7.1+1,-0.3) -- (7.1-1,-0.3) -- cycle;
            \node[] at (7.1, 0) {$C_{i(2)}$};
            \node[] at (9.1, 0) {$\cdots$};
            \filldraw[blue!30, draw=black, line width=1pt] (11.1-1,0.3) -- (11.1+1,0.3) -- (11.1+1,-0.3) -- (11.1-1,-0.3) -- cycle;
            \node[] at (11.1, 0) {$C_{\abs{Q}}$};
            \filldraw[green!30, draw=black, line width=1pt] (2.1-1,-1+0.3) -- (2.1+1.5,-1+0.3) -- (2.1+1.5,-1+-0.3) -- (2.1-1,-1+-0.3) -- cycle;
            \node[] at (2.1+0.25, -1) {$W_1$};
            \node[] at (4.35, -1) {$\cdots$};
            \filldraw[green!30, draw=black, line width=1pt] (6.1-1,-1+0.3) -- (6.1+1,-1+0.3) -- (6.1+1,-1+-0.3) -- (6.1-1,-1+-0.3) -- cycle;
            \node[] at (6.1, -1) {$W_{j}$};
            \filldraw[green!30, draw=black, line width=1pt] (7.9-0.7,-1+0.3) -- (7.9+0.7,-1+0.3) -- (7.9+0.7,-1+-0.3) -- (7.9-0.7,-1+-0.3) -- cycle;
            \node[] at (7.9, -1) {$W_{j+1}$};
            \node[] at (9.15, -1) {$\cdots$};
            \filldraw[green!30, draw=black, line width=1pt] (10.9-1.2,-1+0.3) -- (10.9+1.2,-1+0.3) -- (10.9+1.2,-1+-0.3) -- (10.9-1.2,-1+-0.3) -- cycle;
            \node[] at (10.9, -1) {$W_{\abs{Q}-1}$};
        \end{tikzpicture}
        \caption{Illustration of the repartition using \PropOracle, where $(W_1, \ldots, W_{\abs{Q}-1})$ is a new partition of $P\setminus C_{i(1)}$ and $C_{i(2)}$ only intersects with bundles $W_j$ and $W_{j+1}$.}
    \end{figure}

    Next, consider bundle $C_{i(2)}$. 
    Since $u_i(C_{i(2)})\le u_i(C_{i(1)})$ and all items in $C_{i(2)}$ must be placed continuously by the rearrangement at Line~\ref{alg:prop1_mms_part2pairbegin}, $C_{i(2)}$ can intersect at most $2$ of $(W_1,\ldots,W_{|Q|-1})$. 
    % Since $u_i(C_{i(2)})\le u_i(C_{i(1)})$ and all items in $C_{i(2)}$ must be placed continuously by the rearrangement at Line~\ref{alg:prop1_mms_part2pairbegin} after treating the sequence as a cycle, $C_{i(2)}$ can intersect at most $2$ of $(W_1,\ldots,W_{|Q|-1})$. 
    We then can remove $C_{i(2)}$ and merge the two sets which contain some elements in $C_{i(2)}$. From $u_i(C_{i(2)})\le u_i(C_{i(1)})$, the merged bundle still has a weakly larger utility than $C_{i(1)}$ under $u_i(\cdot)$.
    If $C_{i(2)}$ is contained in exactly one set, we can directly remove that set.
    We then take similar steps for the remaining $C_i$ and conclude the lemma.
\end{proof}

Based on \Cref{lem:prop1_mms_small}, for each agent $i\in Z$, the utility of every allocated bundle is weakly less than the proportional value of the utility of the remaining items.
By repeatedly applying the lemma, we then derive \Cref{lem:prop1_mms_total2}, which states that the proportional value does not decrease.
\begin{lemma}\label{lem:prop1_mms_total2}
For the set $P$ of unallocated items and the set $Q$ of agents not receiving items during the procedure, for any $i\in Q$, we have $ u_i(P)/\abs{Q}\ge u_i(M)/\abs{N}$.
\end{lemma}

% \bt{I cannot see why this lemma holds. We have also allocated singletons to agents at Line~10, right? What if a remaining agent has a large value on an allocated item?}

By putting \Cref{lem:prop1_mms_small} and \Cref{lem:prop1_mms_total2} together, we can conclude that the bundle received by every agent in every round satisfies PROP1 and $1/2$-MMS under the larger instance.

\paragraph{Inductive step.}
Next, we prove the inductive step by showing the correctness of the recursion.
There are two recursive operations within Algorithm~\ref{alg:prop1_mms_general}.
In each of them, we allocate a large bundle (a singleton) to an agent and invoke Algorithm~\ref{alg:prop1_mms_general} on a smaller instance with the removal of this agent and the allocated item.
In \Cref{lem:prop1_mms_notypeone}, we first demonstrate the correctness of the recursion for the remaining agents.
Then in \Cref{lem:correctness_of_fst_item}, we show the allocation is also PROP1 and $1/2$-MMS for agents receiving a large bundle.

\begin{restatable}{lemma}{PropMmsNotypeone}
\label{lem:prop1_mms_notypeone}
    Let $I$ be a problem instance with item set $P$ and agent set $Q$.
    Fix an agent $i'$ and an item $o$.
    For each agent $i\neq i'$, define $\MMS_i$ and $\MMS'_i$ as the maximin share of agent $i$ under the original instance $I$ and the smaller instance $I'$ after removing agent $i'$ and item $o$.
    Assume $\mathcal{A}$ is a PROP1 allocation for the smaller instance $I'$. 
    Then for each $i\in Q\setminus\{i'\}$, we have
    \begin{itemize}[itemsep=0pt]
        \item $\MMS'_i\ge \MMS_i$;
        \item there exists an item $g\in P\setminus A_i$ such that $u_i(A_i\cup\{g\})\ge u_i(P)/\abs{Q}$.
    \end{itemize}
\end{restatable}
\begin{proof}
Fix an agent $i\in N\setminus\{i'\}$.
Assume $\MMS_i$ is obtained from an allocation $\mathcal{B}$.
We construct a new allocation, for instance $I'$, by appending the other items of the bundle containing $o$ to another arbitrary bundle of $\mathcal{B}$, which weakly improves the utility of the least preferred bundle.
Therefore, the first property $\MMS'_i\ge \MMS_i$ holds.

Second, if $u_i(o)\ge \frac{1}{|Q|}\cdot u_i(P)$, then $u_i(A_i\cup\{o\})\ge \frac{1}{|Q|}u_i(P)$ holds directly.
Otherwise, we have $u_i(o)<\frac{1}{|Q|}u_i(P)$, so $u_i(P\setminus\{o\}))>\frac{|Q|-1}{|Q|}u_i(P)$.
Since $\mathcal{A}$ is PROP1 in instance $I'$, by definition, there exists an item $g\in P\setminus (A_i\cup\{o\})$ such that $u_i(A_i\cup\{g\})\ge \frac{1}{|Q|-1}u_i(P\setminus\{o\})>\frac{1}{|Q|}u_i(P)$.
\end{proof}

\begin{restatable}{lemma}{CorrectnessOfFstItem}
\label{lem:correctness_of_fst_item}
    For a bundle $A_k$ of the first type allocated to an agent $k$, $u_k(A_k)\ge u_k(P)/2\ge \MMS_k/2$ and there exists a good $g\in M\setminus A_k$ such that $u_k(A_k\cup\{g\})\ge u_k(M)/ n$.
\end{restatable}
\begin{proof}
    We assume the corresponding sets $P$ and $Q$ when allocating $A_k$ to $k$.
    If bundle $A_k$ is allocated at Lines~\ref{alg:prop1_mms_part1case1begin}-\ref{alg:prop1_mms_part1case1end}, then let $g$ as the second most valuable item in $P'$. 
    We have 
    $$u_k(A_k)\ge \frac{1}{|Q|+|P'|}u_k(P)\ge \frac{1}{2|Q|}u_k(P
    ), \quad u_k(A_k)+u_k(g)\ge \frac{2}{|Q|+|P'|}u_k(P)\ge \frac{1}{|Q|}u_k(P)\,.$$
    From Lemma~\ref{lem:prop1_mms_total2}, we can conclude $u_k(A_k)\ge u_k(M)/2n$ and $u_k(A_k\cup\{g\})\ge u_k(M)/n$. 

    If this bundle $A_k$ is allocated at Lines~\ref{alg:prop1_mms_part2case1begin}-\ref{alg:prop1_mms_part2case1end}, we let $g$ as the second most valuable item in $\{I_1,\ldots,I_{|Q|-1}\}$. By the same analysis as before, we conclude the same statement.
\end{proof}

\paragraph{Query and time complexity.}
Finally, we consider the query and time complexity of Algorithm~\ref{alg:prop1_mms_general}.
After each while-loop iteration, the size of $N$ is decreased by at least $1$. 
% \bt{What does it mean by ``an iteration''? A while-loop iteration? If so, is it possible that Line 10 or 22 is executed after many iterations of the while loop (in which case only one agent is removed after many iterations, and we need to start over again)?}
We consider the query complexity and running time within each iteration.

At Lines~\ref{alg:prop1_mms_part1begin}-\ref{alg:prop1_mms_part1end}, the \ItemPartition subroutine is invoked for an agent and she is asked to find the two most valuable bundles, costing $O(n^2\log m)=O(n^2\log m)$ queries.
In Line~\ref{alg:prop1_mms_part1end}, both the query complexity and running time to construct an EF1 allocation under $u_k$ are $O(n^2\log m)$.
We then consider the second part, corresponding to Lines~\ref{alg:prop1_mms_part2begin}-\ref{alg:prop1_mms_part2end}.
By using the binary search technique in Line~\ref{alg:prop1_mms_part2construction}, it costs $O(n\log m)$ queries for each agent $i$ to construct the bundles from each $C_j$ and $O(n)$ to find the two most valuable bundles in Line~\ref{alg:prop1_mms_part2case1begin}.
Finally, finding the maximal subset $Z$ in Line~\ref{alg:prop1_mms_part3begin} costs $O(n^2)$ and if $|Z|=k$, finding the perfect matching costs $O((n-k)^3)$ in Line~\ref{alg:prop1_mms_part3perfect}, and both of them invoke no query.
Then, the running time $T(n,m)$ is given by
\begin{align*}
T(n, m)&=T(n-1, m)+O(n^2\log m)+O(n^2\log m)+O(n^3\log m)+O(n^2)+O((n-k)^3)\\
&=T(n-1, m)+ O(n^3\log m)= O\left(n^4\log m\right).
\end{align*}
Similarly, the query complexity $Q(n,m)$ is given by
\begin{align*}
Q(n, m) &= Q(n-1, m)+O(n^2\log m)+O(n^2\log m)+O(n^3\log m) = O\left(n^4\log m\right),
\end{align*}
which concludes the proof.
% By considering the number of the large bundles, which is $O(n)$, the query and time complexity are both $O(n^5\log m)$.
\end{proof}

\section{Conclusion and Future Work}
In this work, we proposed the comparison-based query model, where the valuation functions are not provided as the input to the allocation algorithms.
Our model provides a new approach to access the agents' preferences for items.
A highlight of our algorithms is that our results only rely on the comparison result between bundles, rather than the concrete values.
Moreover, we also studied the query complexity of computing a fair allocation and provided an efficient algorithm to compute an allocation satisfying both PROP1 and $\frac12$-MMS through at most $O(\log m)$ queries for a constant number of agents.
We believe this will provide inspiration for designing algorithms based on comparative relationships.

One of our future directions is to proceed the research under the EF1 notion for a constant number of agents and additive valuations.
Another interesting future direction is to consider general valuation functions (e.g., submodular valuations and monotone valuations) beyond additive ones.

\section*{Acknowledgments}
The research of Biaoshuai Tao was supported by the National Natural Science Foundation of China (No. 62102252 and 62472271). 
% \shengxin{Biaoshuai should include the new NSFC project number.}\bt{done}
The research of Shengxin Liu was partially supported by the National Natural Science Foundation of China (No. 62102117), by the Shenzhen Science and Technology Program (No. GXWD20231129111306002), by the Guangdong Basic and Applied Basic Research Foundation (No. 2023A1515011188), and by the Key Laboratory of Interdisciplinary Research of Computation and Economics (Shanghai University of Finance and Economics), Ministry of Education.

\bibliographystyle{alpha}
\bibliography{reference}

\appendix
\newpage

\section{Omitted Proofs in Sect.~\ref{sec:prop1}}
    
\subsection{An Alternative Approach to Defining \texttt{PROP1-PROP-Subroutine}}
\label{appendix:alt_prop1_prop_subroutine}

The primary algorithm for the subroutine is presented in Algorithm~\ref{alg:oracle_alt}.
Specifically, we first execute Algorithm~\ref{alg:prop1_identical} where we initialize the input number of the agents to be $n-1$ and the input item set to be $M\setminus B$.
The algorithm returns a PROP1 allocation with $n-1$ bundles, and we find the smallest bundle in the allocation, say $A_1$.
We output `yes' if $u(A_1)\le u(B)$ and `no' if $u(A_1)>u(B)$.

\begin{algorithm}[h]
\caption{\texttt{PROP1-PROP-Subroutine-ALT}$(n, M, B, u)$}\label{alg:oracle_alt}
Execute Algorithm~\ref{alg:prop1_identical} where the input is set to be $n-1$ and $M\setminus B$, and obtain allocation $(A_1, \dots, A_{n-1})$\;
Let $i\leftarrow \argmin\limits_{k\in[n-1]} u(A_k)$ be the index of the smallest bundle among $A_1,\dots,A_{n-1}$\;
\eIf{$u(A_i)\le u(B)$}{
    \Return{yes}\;
}{
    \Return{no}\;
}
\end{algorithm}

\paragraph{Correctness of Algorithm~\ref{alg:oracle_alt}.}
Both the query complexity and running time of Algorithm~\ref{alg:oracle} are $O\left(n^2\log m\right)$, the same as in \Cref{thm:prop1_identical}.
For simplicity, we assume $u(M)=1$, $u(B)=x$ and $u(A_1)=y$.
According to Theorem~\ref{thm:prop1_identical}, $(A_1, \dots, A_{n-1})$ is PROP1 for $M\setminus B$ and $n-1$ agents, hence, there exists an item $g\in M\setminus (A_1 \cup B)$ such that $y+u(g)\ge \frac{1-x}{n-1}$. 
If the algorithm outputs `yes', $y\le x$, then
$$x+u(g)\ge \frac{1-x}{n-1}\ \Rightarrow\ x\ge \frac1n - \frac{n-1}{n}u(g)\ \Rightarrow\ x+u(g)\ge \frac1n.$$
Hence, bundle $B$ satisfies PROP1.
If the algorithm outputs `no', then $y>x$ and
$B$ is non-PROP.
Otherwise, 
$$u(M)=\sum_{i\in [n-1]}u(A_i)+u(B) \ge (n-1)\cdot y+x> nx \ge 1,$$
which contradicts to $u(M)=1$.
\qed

\section{Omitted Proofs in Sect.~\ref{sec:ef1}}
\label{appendix:ef1}
\subsection{Proof of Theorem~\ref{thm:2agents_additive_ef1}}
\label{appendix:proof_of_ef1_for_two_agents}
\TwoAgentsAddEFOne*

\begin{proof}
For two agents, a straightforward algorithm \emph{cut-and-choose} can output an EF1 allocation. 
In the \emph{cut-and-choose} algorithm, the first agent divides all the items into two bundles, where she perceives the allocation as EF1 regardless of which bundle she receives.
Then, the second agent chooses a bundle that she feels is more valuable.
We can easily verify that the allocation is EF1 to the first agent and envy-free to the second agent.
The query complexity mainly depends on the first agent's division step, as the second agent's choosing step only requires one query.

For the division step, we claim it can be achieved in $O(\log m)$ queries by adopting the binary search method.
We arrange all the items in a line.
For any given item $g$, we define the bundle to the left of $g$ as $L(g)$ ($g$ is not included in $L(g)$), and the bundle to the right of $g$ as $R(g)$ ($g$ is not included in $R(g)$).
Our goal is to find the rightmost item $g$ such that $u_1(L(g))<u_1(R(g))$.
% and $v_1(L(g)\cup \{g\})\ge v_1(R(g))$.
It is straightforward that such an item $g_0$ exists due to the interpolation theorem, and can be found through a binary search algorithm.
Then a valid partition for the first agent is $(L(g_0)\cup\{g_0\}, R(g_0))$.

Therefore, the overall query complexity and running time are $O(\log m)$.
\end{proof}

\subsection{Proof of Theorem~\ref{thm:3agents_additive_ef1}}
\label{sec:detailed_proof_of_ef1_3agents}
\EfThree*
Our algorithm works in the following steps:

\medskip
\noindent \textbf{Step 1:} Compute an EF1 allocation $(A, B, C)$ with identical valuation function $u_1$ according to the algorithm described in Theorem~\ref{thm:ef1_identical}.
If agent $2$ and agent $3$ favorite different bundles, then assign the remaining bundle to agent $1$.
The current allocation is clearly EF1.
Otherwise, without loss of generality, we assume $u_2(A) > u_2(B) \ge u_2(C)$ and $u_3(A) > \max\{ u_3(B), u_3(C)\}$.

\medskip
\noindent \textbf{Step 2:} Next we arrange items of $A$ on a line and use binary search to divide bundle $A$ into $A'$ and $T$ such that $u_2(A') \le u_2(B)$ and $u_2(A'\cup \{g_t\}) > u_2(B)$, where $g_t$ is the leftmost item of $T$.
\begin{itemize}
    \item If $u_3(A') \ge \max\{u_3(B), u_3(C)\}$, then we allocate $A'$ to agent $3$, $B$ to agent $2$ and $C$ to agent $1$.
    For item set $T$ and identical valuation function $u_2$, we could compute an EF1 allocation $(T_1, T_2, T_3)$.
    Then we let the three agents select their preferred bundle in order $3\rightarrow 1 \rightarrow 2$.
    \item Otherwise, if $u_3(A') < \max\{u_3(B), u_3(C)\}$, we proceed to the next step.
\end{itemize}

\medskip
\noindent \textbf{Step 3:}  We call an item $g$ \emph{large} if $g\in T$ and $u_2\left(A'\cup \{g\}\right) \ge u_2(B)$.
Then we will find three large items (if exist) as follows:

First, since $g_t$ obviously satisfies the definition of a large item, it is the first one.

Next, we use binary search to find a set $E\subseteq T$ such that $u_2(A'\cup E) \ge u_2(B)$ and $u_2\left(A'\cup E \setminus \{g_e\}\right) < u_2(B)$, where $g_e$ is the rightmost item of set $E$.
\begin{itemize}
    \item If such set $E$ and item $g_e$ exist, then remove the items of $E\setminus \{g_e\}$ from $T$ and add them to $A'$.
    Thus, $g_e$ is a new large good.
    If $u_3(A') \ge \max\{u_3(B), u_3(C)\}$, then we return to step 2 with the updated bundle $A'$.
    Otherwise, the condition $u_3(A') < \max\{u_3(B), u_3(C)\}$ still holds.
    We could continue the above operation until we find three large items.
    \item Otherwise, if $E$ and $g_e$ do not exist, we remove all items except the (up to two) \emph{large} items from $T$ and add them to $A'$.
\end{itemize}

\medskip
\noindent \textbf{Step 4:} Let $S_3$ be agent $3$'s favourite bundle between $B$ and $C$, and let $S_1$ be the other bundle.
Then assign $A'$ to agent $2$ and $S_3$ to agent $3$ and $S_1$ to agent 1.

\medskip
\noindent \textbf{Step 5:} Compute an EF1 allocation $(T_1', T_2', T_3')$ with identical valuation function $u_3$ according to the algorithm in Theorem~\ref{thm:ef1_identical}, where the three large items are excluded from the item set.
Denote the three large items by $g_t,g_e,g_r$, and without loss of generality, assume $u_3(g_1)\ge u_3(g_2)\ge u_3(g_3)$, $u_3(T_1')\le u_3(T_2')\le u_3(T_3')$.
Then we obtain an allocation $(T_1'\cup \{g_1\}, T_2'\cup \{g_2\}, T_3'\cup \{g_3\})$ and denote it by $(T_1, T_2, T_3)$.

\medskip
\noindent \textbf{Step 6:} Consider the following cases of the locations of the large items:
\begin{itemize}
    \item If there exists a large item in each of $T_1$, $T_2$, and $T_3$, then let the three agents take their favorite bundle according to the order $2\rightarrow 1 \rightarrow 3$.
    \item Otherwise, that means there are at most two items within $T_1 \cup T_2 \cup T_3$.
    We assign the first large item to agent $2$ and the second one to agent $1$.
\end{itemize}

\begin{proof}[Proof of Theorem~\ref{thm:3agents_additive_ef1}]
To prove the allocation is EF1, we only need to show $(T_1, T_2, T_3)$ after Step~5 is EF1 under the identical valuation function $u_3$, and the proof in the rest of the part is the same as in~\cite{Oh_2021}.
As $(T_1',T_2',T_3')$ satisfies EF1, we know for any $i,j \in[3]$,
there exists an item $g\in T_j$ such that $u_3(T_i')\ge u_3(T_j'\setminus\{g\})$.
If $i>j$, $u_3(T_i)\ge u_3(T_i')\ge u_3(T_j')=u_3(T_j\setminus\{g_j\})$.
If $i<j$, $u_3(T_i)=u_3(T_i')+u_3(g_i)\ge u_3(T_j'\setminus\{g\})+u_3(g_j)=u_3(T_j\setminus\{g\})$.
Hence we conclude $(T_1, T_2, T_3)$ satisfies EF1.

As the number of agents is constant, the binary search in Step~2 and Step~3 requires $O(\log m)$ queries.
The query complexity to find an EF1 allocation under identical valuation is also $O(\log m)$ according to \Cref{thm:ef1_identical}.
The other operations in the above algorithm cost $O(1)$ queries.
Hence, the query complexity of computing an EF1 allocation for three agents is $O(\log m)$.
\end{proof}

\end{document}